\documentclass[12pt,a4paper]{article}

\usepackage[english]{babel}

\usepackage[a4paper,top=2cm,bottom=2cm,left=2.5cm,right=2.5cm,marginparwidth=1.75cm]{geometry}

\usepackage[style=authoryear-comp, backend=biber, sortcites=false]{biblatex}
\usepackage{amsmath}
\allowdisplaybreaks
\usepackage{amssymb}
\usepackage{amsthm}
\usepackage{bbm}
\usepackage{graphicx}
\usepackage[colorlinks=true, allcolors=blue]{hyperref}
\usepackage[title]{appendix}
\usepackage{mathrsfs}
\usepackage{amsfonts}
\usepackage{caption}
\usepackage{authblk}
\usepackage{csquotes}
\usepackage[T1]{fontenc} 
\usepackage{lipsum}
\usepackage{float}

\usepackage{setspace}
\usepackage{titlesec}
\titleformat{\section} 
  {\normalfont\Large\bfseries}{\thesection.}{1em}{}
  
\theoremstyle{plain}
\newtheorem{theorem}{Theorem}
\newtheorem{lemma}[theorem]{Lemma}
\newtheorem{proposition}[theorem]{Proposition}

\newtheorem{assumption}{Assumption}

\theoremstyle{definition}
\newtheorem{definition}{Definition}
\newtheorem{example}{Example}

\title{Selection of Efficient Monetary Equilibria Through Aggregate Real Savings-Based Taylor Rule}

\author[1,2]{Leandro Lyra Braga Dognini}
\affil[1]{\small Department of Economics, Rio de Janeiro State University}
\affil[2]{\small Legislative Advisory, Federal Senate of Brazil}
\date{\today} 

\begin{document}
\maketitle
\begin{abstract}
\noindent This paper uses the generalized Cass criterion $\sum^{\infty}_{t=1}(\Vert p_{t}\Vert\sum_{h\in G_{t}}\Vert e^{h}_{t}\Vert)^{-1}=\infty$ to extend the results from \textcite{Dognini_2026a} regarding the existence of efficient monetary equilibria on consumption-loan overlapping generations economies. These results reveal that if the economy is prone to savings, then monetary equilibria will emerge in a pure non-stationary general equilibrium model with heterogeneous households, thus providing a solution to the \textcite{Hahn_1965} problem. It is also proved that, in prone-to-savings economies, non-vanishing relative aggregate real savings fully characterize efficient monetary equilibria. I use this result to show that a Taylor rule based on an inflation ceiling and a relative aggregate real savings floor can be used to control the price level and lead the economy towards an efficient monetary equilibrium. In contrast, a Taylor rule based solely on an inflation target is able to control the price level but generally leads the economy towards an inefficient monetary equilibrium. 
\end{abstract}



\section{Introduction}\label{sec1}

\textsc{The \textcite{Hahn_1965} problem} (e.g., \textcite[p. 1486]{Bewley_1983}) poses the question of what are the conditions under which a monetary equilibrium (i.e., an equilibrium in which fiat money has value) exists in a general equilibrium model. 

There are several solutions to this problem in the literature based on cash-in-advance restrictions (e.g., \textcite{Clower_1967, Okuno_1978, Santos_1990}), transaction costs (e.g., \textcite{Kurz_1974, OstroyStarr_1974}), taxes (e.g., \textcite{Starr_1974, Starr_2003}), stationary economies (e.g., \textcite{Lucas_1972,Grandmont_1973,Hayashi_1976, RaoAiyagari_1987, Esteban_1994,Burke_1999}) or low-dimensional models (e.g., \textcite{Wallace_1980}).

In a recent paper \parencite{Dognini_2026a}, I have shown that in a general non-stationary consumption-loan overlapping generations model (e.g., \textcite{Samuelson_1958}) (without durable, dividend-paying assets, cash-in-advance constraints, wealth-transfer mechanisms, or transaction costs), if the economy is prone to savings, then there exists an efficient monetary equilibrium. The proof of this existence result relies on the \textcite{Cass_1972} criterion stated in the following form \parencite[pp. 1932-1936]{GeanakoplosPolemarchakis_1991}
\begin{eqnarray}\label{eqCass1}
    \sum^{\infty}_{t=1}\frac{1}{H_{t}\Vert p_{t}\Vert}=+\infty.
\end{eqnarray}

This form, however, assumes that endowments are uniformly bounded. A more general form of the Cass criterion can be achieved if one uses the inequalities associated with the deep projections of the canonical manifolds of consumer theory \parencite{Dognini_2025b,Dognini_2026c}, leading to
\begin{eqnarray}\label{eqCass2}
    \sum^{\infty}_{t=1}\frac{1}{\Vert p_{t}\Vert\sum_{h\in G_{t}}\Vert e^{h}_{t}\Vert}=+\infty.
\end{eqnarray}

The first contribution of this paper, therefore, is to generalize the definitions and results in \textcite{Dognini_2026a} using the Cass criterion form (\ref{eqCass2}) instead of (\ref{eqCass1}). Of particular relevance is the more general definition of prone-to-savings economies (see Definition \ref{defProneSavings} here and Definition 5 in \textcite[p. 6]{Dognini_2026a}), which highlights that, in these economies, low relative aggregate real savings can only occur if the real return rates are small relative to the growth rate of first-period endowments.

Reaching more general results on efficient monetary equilibria is, by itself, an important milestone for properly deploying overlapping generations models, which are used for a wide range of economic applications other than monetary theory itself, stemming from the literature on bubbles (e.g., \textcite{Tirole_1985,SantosWoodford_1997,HiratoToda_2025}) to the one on social security design (e.g., \textcite{Samuelson_1975,Dognini_2026b}). 

The second contribution of this paper is related to the study of interest rate rules (e.g., \textcite[pp. 37-49]{Woodford_2003}) and, more specifically, to that of \textcite{Taylor_1993} rules (i.e., feedback rules that set nominal interest rates as a function of inflation rates rather than the price level itself). 

A central matter of Taylor rules is their relation to equilibrium indeterminacy. This is particularly important in overlapping generations economies since equilibria are generally not locally unique and, therefore, indeterminacy is a \textit{structural} feature of these models (e.g., \textcite{KehoeLevine_1984_1,KehoeLevine_1984_2,KehoeLevine_1985,KehoeLevine_1990,MullerWoodford_1988,KehoeLevineMas-ColellZame_1989,Burke_1990, KehoeLevineMas-ColellWoodford_1991}). 

With this indeterminacy matter in mind, I first show that, in the absence of any interest rate rule, discretionary monetary policy increases the level of indeterminacy of these economies when measured by the cardinality of the respective sets of equilibria (see Proposition \ref{propCardinalitySetsEquilibria}) and poses no restriction on the real outcomes of the economy (e.g., \textcite{1975_SargentWallace}). 

Also, as is common in the literature (e.g., \textcite{Clarida_2000,Benhabib_2001,Cochrane_2011}), I show that Taylor rules act as a \textit{selection mechanism} that excludes equilibria incompatible with it, thus leaving the economy in some of the \textit{stable equilibria} relative to it (see Definition \ref{defTaylorRule}).

Then, I study the set of stable equilibria relative to different Taylor rules and the main result can be summarized as follows: under a Taylor rule based on a not-too-low \textit{inflation ceiling} and a not-too-high \textit{relative aggregate real savings floor}, inflation is bounded and the economy is led towards an efficient monetary equilibrium. In contrast, a Taylor rule based on a not-too-low \textit{inflation target} is also able to bound inflation and stabilize the economy, but generally leads towards inefficient equilibria. 

After this \hyperref[sec1]{Introduction}, the outline of the paper is as follows. \hyperref[sec2]{Section 2} extends the definitions and results from \textcite{Dognini_2026a} based on the generalized \textcite{Cass_1972} criterion (\ref{eqCass2}). \hyperref[sec3]{Section 3} discusses the different Taylor rules. Finally, a brief
discussion of the results is presented in the \hyperref[sec4]{Conclusion} 
and all omitted proofs are stated in the \hyperref[appx]{Appendix}.


\section{Results on the existence of optimal monetary equilibria based on the generalized Cass criterion}\label{sec2}

This section extends the definitions and results in \textcite{Dognini_2026a} based on the generalized statement of the \textcite{Cass_1972} criterion (\ref{eqCass2})\footnote{See \textcite{Dognini_2026c} for a detailed discussion and proof of this form of the \textcite{Cass_1972} criterion.}.

The economy $\mathcal{E}$ is a consumption-loan overlapping generations one with discrete time periods $t\in\mathbb{N}_{0}$ and $L_{t}\in\mathbb{N}$ perishable commodities in each period. Households live for two periods, the one they are born into and the next, are indexed by $h\in \mathbb{N}$ and are gathered in generations $G_{t}=\{h\in\mathbb{N}\mid h \textrm{ is born in period }t\}$, $t\geq0$, according to their period of birth.

Household $h\in G_{t}$ is defined through a consumption set $X^{h}=\mathbb{R}^{L_{t}+L_{t+1}}_{+}$, a nonzero endowment $e^{h}=(e^{h}_{t},e^{h}_{t+1})\in\mathbb{R}^{L_{t}+L_{t+1}}_{+}$ and a continuous, non-decreasing, semi-strictly quasiconcave\footnote{Utility function $u(\cdot)$ is \textit{semi-strictly quasiconcave} (see \textcite[p. 87, Definition 2]{ArrowHahn_1971}), if $u(\tilde{c})\geq u(c)$ implies $u(\alpha\tilde{c}+(1-\alpha) c)\geq u(c)$, $0\leq \alpha \leq 1$, and if $u(\tilde{c})>u(c)$ implies $u(\alpha\tilde{c}+(1-\alpha) c)>u(c)$, $0< \alpha \leq1$. In particular, if $u(\cdot)$ is semi-strictly quasiconcave, then it is quasiconcave.} utility function $u^{h}:\mathbb{R}^{L_{t}+L_{t+1}}_{+}\rightarrow\mathbb{R}$ without local maxima. Also, $\sum_{t\geq0}\sum_{h\in G_{t}}e^{h}\in\mathbb{R}^{\infty}_{++}$ and $\alpha_{t}=\sum_{h\in G_{t+1}}\Vert e^{h}_{t+1}\Vert/\sum_{h\in G_{t}}\Vert e^{h}_{t+1}\Vert$\footnote{For $c\in\mathbb{R}^{L}$, $L\geq1$, $\Vert c\Vert=\sum^{L}_{i=1}\vert c_{i}\vert$.}, $t\geq0$, represents the intraperiod growth rate between young-age and old-age aggregate endowments faced by generation $G_{t}$ when old. 

I proceed with the following assumption.

\begin{assumption}[Bounded intraperiod growth rates]\label{assBoundsPopAndEndownments}
There are $0<\alpha_{\min}<\alpha_{\max}$ such that $\alpha_{t}\in(\alpha_{\min},\alpha_{\max})$, $t\geq0$.
\end{assumption}

Assumption \ref{assBoundsPopAndEndownments} requires the existence of strictly positive uniform bounds on the intraperiod growth rates of aggregate endowments. The equilibrium existence result of \textcite[p. 119, Theorem 5]{ArrowHahn_1971} leads to the next assumption.

\begin{assumption}\label{assResourceRelated}
For $j\geq 0$, all households $h\in \bigcup^{j}_{t=0}G_{t}$ are indirectly resource related \footnote{Following \textcite[117-118]{ArrowHahn_1971}, given $j\geq0$, household $h^{\prime}\in\bigcup^{j}_{t=0}G_{t}$ is \textit{resource related} to household $h^{\prime\prime}\in\bigcup^{j}_{t=0}G_{t}$ if, for every allocation $\{c^{h}\}_{h\in\bigcup^{j}_{t=0}G_{t}}$, $c^{h}\in\mathbb{R}^{L_{t}+L_{t+1}}_{+}$, $h\in G_{t}$, $0\leq t\leq  j$, with $\sum_{0\leq t\leq  j}\sum_{h\in G_{t}}c^{h}\leq \sum_{0\leq t\leq  j}\sum_{h\in G_{t}}e^{h}$, there exists an allocation $\{\tilde{c}^{\, h}\}_{h\in\bigcup^{j}_{t=0}G_{t}}$ and $y\in\mathbb{R}^{\sum^{j+1}_{i=0}L_{i}}_{+}$ such that $\sum_{0\leq t\leq  j}\sum_{h\in G_{t}}\tilde{c}^{\, h}\leq\sum_{0\leq t\leq  j}\sum_{h\in G_{t}}e^{h} +y$,
\begin{eqnarray*}
        u^{h}(\tilde{c}^{\, h})&\geq&u^{h}(c^{h}),\textrm{ for all }h\in\bigcup^{j}_{t=0}G_{t},\\
        u^{h^{\prime\prime}}(\tilde{c}^{\, h^{\prime\prime}})&>&u^{h^{\prime\prime}}(c^{h^{\prime\prime}}),
\end{eqnarray*}
and
\begin{eqnarray*}
    y_{ti}>0\textrm{ only if } e^{h^{\prime}}_{ti}>0,
\end{eqnarray*}
for $i\in\{1,\ldots, L_{t}\}$, $0\leq t\leq j+1$. Household $h^{\prime}$ is \textit{indirectly resource related} to household $h^{\prime\prime}$ if there is a sequence of households $\{h_{m}\}_{0\leq m\leq n}$, $n\geq1$, with $h_{0}=h^{\prime}$, $h_{n}=h^{\prime\prime}$ and $h_{m}$ resource related to $h_{m+1}$, $0\leq m\leq n-1$. 
}.
\end{assumption}

Assumption \ref{assBoundsPopAndEndownments} is a technical one that poses no significant constraint in the model. Household $h\in G_{t}$, $t\geq0$, has a Walrasian demand function $x^{h}:\mathbb{R}^{L_{t}+L_{t+1}}_{++}\rightarrow \mathbb{R}^{L_{t}+L_{t+1}}_{+}$, $x^{h}(p_{t},p_{t+1})=(x^{h}_{t}(p_{t},p_{t+1}),x^{h}_{t+1}(p_{t},p_{t+1}))$, $(p_{t},p_{t+1})\in\mathbb{R}^{L_{t}+L_{t+1}}_{++}$, given by 
\begin{eqnarray}\label{eqUMP}
    x^{h}(p_{t},p_{t+1})=\textrm{argmax}_{c\in\mathbb{R}^{L_{t}+L_{t+1}}_{+}}& u^{h} (c) \\
    \textrm{ s.t. }& (p_{t},p_{t+1})\cdot (c-e^{h})\leq0.\nonumber
\end{eqnarray}
The excess demand function in period $t\geq1$, $z_{t}:\mathbb{R}^{L_{t-1}+L_{t}+L_{t+1}}_{++}\rightarrow \mathbb{R}^{L_{t}}$, is 
\begin{eqnarray*}
    z_{t}(p_{t-1},p_{t},p_{t+1})=\sum_{h\in G_{t-1}}(x^{h}_{t}(p_{t-1},p_{t})-e^{h}_{t})+\sum_{h\in G_{t}}(x^{h}_{t}(p_{t},p_{t+1})-e^{h}_{t}),
\end{eqnarray*}
and joint excess demand until period $t\geq1$, $Z_{t}:\mathbb{R}^{\sum^{t+1}_{i=0}L_{i}}_{++}\rightarrow\mathbb{R}^{\sum^{t}_{i=1}L_{i}}$, is $Z_{t}(p_{0},\ldots,p_{t+1})=(z_{1}(p_{0},p_{1},p_{2}),\ldots,z_{t}(p_{t-1},p_{t},p_{t+1}))$. Before stating the third assumption, let $\beta=\max\, \{1+\alpha^{-1}_{\min},1+\alpha_{\max}\}$ and $\mathcal{B}_{t}(\sigma)=\{p\in\mathbb{R}^{L_{t}+L_{t+1}}_{++}\mid \sigma \leq p_{i}\leq \sigma^{-1}, 1\leq i \leq L_{t}+L_{t+1}\}$, for $\sigma\in(0,1)$, $t\geq 0$. It follows that the ``box-set'' $\mathcal{B}_{t}(\sigma)\subset \mathbb{R}^{L_{t}+L_{t+1}}_{++}$ is compact. 
\begin{assumption}[Unbounded aggregate demand at border prices]\label{assBoundsPrices}
For all $t\geq0$, there is $\sigma_{t}\in(0,1)$ such that if $(p_{t}/p_{t1},p_{t+1}/p_{t1})\notin \mathcal{B}_{t}(\sigma_{t})$, then $ \sum_{h\in G_{t}} \Vert x^{h}(p_{t},p_{t+1})\Vert> \beta \sum_{h\in G_{t}}\Vert e^{h}\Vert$, for $(p_{t},p_{t+1})\in \mathbb{R}^{L_{t}+L_{t+1}}_{++}$.
\end{assumption}

Assumption \ref{assBoundsPrices} states that if the prices faced by generation $G_{t}$ move sufficiently close to the border of the corresponding simplex, then its aggregate demand will extrapolate the existing resources on some of the commodities of periods $t$ or $t+1$, considering all possible intraperiod growth trends compatible with Assumption \ref{assBoundsPopAndEndownments}. It is also convenient to define the following compact sets $\mathcal{K}_{t}\subset\mathbb{R}^{L_{t}}_{++}$, $t\geq0$, and $\mathcal{K}\subset\mathbb{R}^{\infty}_{++}$,
\begin{eqnarray*}
    \mathcal{K}_{0}&=&\{p\in\mathbb{R}^{L_{0}}_{++}\mid \sigma_{0}\leq p_{i}\leq \sigma_{0}^{-1}, 1\leq i\leq L_{0}\}\\
    \mathcal{K}_{t}&=&\biggr\{p\in\mathbb{R}^{L_{t}}_{++}\mid \prod^{t-1}_{j=0}\sigma_{j}\leq p_{i}\leq \biggr(\prod^{t-1}_{j=0}\sigma_{j}\biggr)^{-1}, 1\leq i\leq L_{t}\biggr\},
\end{eqnarray*}
for $t\geq1$, and $\mathcal{K}=\prod^{+\infty}_{t=0}\mathcal{K}_{t}$. Throughout this paper, as a topological space, $\mathbb{R}^{\infty}$ is endowed with the product topology.

I proceed with the following definition.

\begin{definition}\label{defSetEquilibria}
The \textit{set of equilibria} $\mathcal{H}\subset\mathbb{R}^{\infty}_{++}$ of economy $\mathcal{E}$ is
\begin{eqnarray*}
\mathcal{H}=\{p\in\mathbb{R}^{\infty}_{++}\mid p_{01}=1,(p_{0},p_{1})\in\mathcal{B}_{0}(\sigma_{0})\textrm{ and } z_{t}(p_{t-1},p_{t},p_{t+1})=0, t\geq 1\}.
\end{eqnarray*}
\end{definition}

Definition \ref{defSetEquilibria} states that all equilibrium price sequences $p=(p_{0},p_{1},\ldots)\in\mathbb{R}^{\infty}_{++}$ are defined in terms of the first good in period $t=0$ and market clearing holds in all periods $t\geq1$. In period $t=0$, however, Definition \ref{defSetEquilibria} only requires the aggregate demand of generation $G_{0}$ to be compatible with Assumptions \ref{assBoundsPopAndEndownments} and \ref{assBoundsPrices}\footnote{See the discussion on \textcite[Section 3, pp. 3-5]{Dognini_2026a} on the differences between this definition of the set of equilibria and the one from classical models (e.g., \textcite{CassOkunoZilcha_1979}, \textcite{BalaskoShell_1980,BalaskoShell_1981_1}, \textcite{OkunoZilcha_1980}, \textcite{KehoeLevine_1985} and \textcite{RaoAiyagari_1992}). While the latter assumes the inception of the economy occurs in the first period of the model, the former assumes long-standing economic activity.}.

\begin{theorem}\label{theoEqSetCompact}
Under Assumptions \ref{assBoundsPopAndEndownments}--\ref{assBoundsPrices}, the set of equilibria $\mathcal{H}\subset\mathbb{R}^{\infty}_{++}$ is non-empty and compact. 
\end{theorem}

We are particularly interested in equilibrium prices $p\in\mathcal{H}$ whose allocation cannot be Pareto dominated, meaning that there is no possible utility-enhancing redistribution of resources. This leads to the following definition\footnote{To ease notation, I consider vectors in $\mathbb{R}^{L_{t}+L_{t+1}}$, $t\geq0$, as elements of $\mathbb{R}^{\infty}$, with the appropriate embedding.}.
\begin{definition}\label{defParetoOptimalEqSet}
The \textit{subset of Pareto optimal equilibria} $\mathcal{H}^{PO}\subseteq\mathcal{H}$ is
\begin{eqnarray*}
    \mathcal{H}^{PO}=\{p\in\mathcal{H}\mid \nexists \{y^{h}\}_{h\geq1}, y^{h}\in\mathbb{R}^{L_{t}+L_{t+1}}_{+}, h\in G_{t},t\geq0, \textrm{ s.t. } \sum_{t\geq0}\sum_{h\in G_{t}}(y^{h}-x^{h}(p_{t},p_{t+1}))=0\\
    \textrm{and }u^{h}(y^{h})\geq u^{h}(x^{h}(p_{t},p_{t+1})), h\in G_{t}, t\geq0, \textrm{ with at least one strict inequality}\}.
\end{eqnarray*}
\end{definition}

The next assumption requires the subset of Pareto optimal equilibria $\mathcal{H}^{PO}\subseteq\mathcal{H}$ to be fully characterized by the generalized \textcite{Cass_1972} criterion (\ref{eqCass2})\footnote{All the existence results in this section are, in fact, results on the existence of equilibrium price sequences that satisfy the \textcite{Cass_1972} criterion (\ref{eqCass2}) (see the proof of Theorem \ref{theoClassModelExistence}, where this point is raised). Therefore, Assumption \ref{assCassCriterion} could require only the sufficiency of the Cass criterion for optimality (e.g., when the conditions of Theorem 3A from \textcite[p. 803]{OkunoZilcha_1980} are satisfied), and this would still ensure the existence of \textit{optimal monetary equilibria}. However, in this case, the set of equilibria that satisfy (\ref{eqCass2}) would be a subset of $\mathcal{H}^{PO}$. Furthermore, if this sort of assumption is fully dropped, then the results in Section \ref{sec3} still furnish the existence of \textit{monetary equilibria}.}.

\begin{assumption}[Generalized Cass criterion]\label{assCassCriterion}
    Let $p\in\mathcal{H}$. Then, $p\in\mathcal{H}^{PO}$ if, and only if, 
    \begin{eqnarray*}
         \sum^{+\infty}_{t=0}\frac{1}{\Vert p_{t}\Vert \sum_{h\in G_{t}}\Vert e^{h}_{t}\Vert}=+\infty.
    \end{eqnarray*}
\end{assumption}

I proceed with the following definition from \textcite{Dognini_2026a}.

\begin{definition}\label{defRealSavings}
The \textit{real savings} $s^{h}:\mathbb{R}^{L_{t}+L_{t+1}}_{++}\rightarrow \mathbb{R}$ of household $h\in G_{t}$, $t\geq0$, is
    \begin{eqnarray*}
       s^{h}(p_{t},p_{t+1}) = \frac{p_{t}}{\Vert p_{t}\Vert}\cdot (e^{h}_{t}-x^{h}_{t}(p_{t},p_{t+1})).
    \end{eqnarray*}
\end{definition}

Definition \ref{defRealSavings} states that the ratio between first-period nominal savings and the cost of the reference basket $(1,\ldots,1)\in\mathbb{R}^{L_{t}}$ (i.e. $p_{t}\cdot(1,\ldots,1)=\Vert p_{t}\Vert$) is a real measure of savings of household $h\in G_{t}$, $t\geq0$. The following lemma reveals a first relation between aggregate real savings and the subset of Pareto optimal equilbria $\mathcal{H}^{PO}$.

\begin{lemma}\label{lemmaFirstParetoOptEquilibria}
Under Assumptions \ref{assBoundsPopAndEndownments}, \ref{assBoundsPrices} and \ref{assCassCriterion}, if $p\in\mathcal{H}/\mathcal{H}^{PO}$, then 
\begin{eqnarray*}
    \lim_{t\rightarrow\infty}\frac{\sum_{h\in G_{t}}s^{h}(p_{t},p_{t+1})}{\sum_{h\in G_{t}} \Vert e^{h}_{t}\Vert}=0.
\end{eqnarray*}
\end{lemma}

Lemma \ref{lemmaFirstParetoOptEquilibria} states that equilibrium prices that are not in $\mathcal{H}^{PO}$ (i.e., inefficient ones) make the ratio of aggregate real savings to aggregate first-period endowment approach zero. In this scenario, the fraction of the aggregate endowment that is being saved by the younger generations (and, thus, being transferred as ``excess consumption'' to the older generations) becomes negligible over time, and autarky or intraperiod trade is all that remains in the long run. In order to have the converse of Lemma \ref{lemmaFirstParetoOptEquilibria}, I proceed with the following definition.

\begin{definition}\label{defProneSavings}
The economy $\mathcal{E}$ is \textit{prone to savings} if there are $\varepsilon>0$, $\delta>0$ such that, for $t\geq0$, $(p_{t},p_{t+1})\in\mathbb{R}^{L_{t}+L_{t+1}}_{++}$ and $(p_{t}/p_{t1},p_{t+1}/p_{t1})\in\mathcal{B}_{t}(\sigma_{t})$,
\begin{eqnarray*}
    \frac{\sum_{h\in G_{t}} s^{h}(p_{t},p_{t+1})}{\sum_{h\in G_{t}}\Vert e^{h}_{t}\Vert}\leq \delta\implies \frac{\Vert p_{t} \Vert}{\Vert p_{t+1} \Vert} \leq \frac{1}{1+\varepsilon}\frac{\sum_{h\in G_{t+1}}\Vert e^{h}_{t+1}\Vert}{\sum_{h\in G_{t}}\Vert e^{h}_{t}\Vert}.
\end{eqnarray*}
\end{definition}

Definition \ref{defProneSavings} states the fundamental condition to be imposed on overlapping generations economies to reach our envisioned existence results. It says that, in prone-to-savings economies, low aggregate real savings $\sum_{h\in G_{t}} s^{h}(p_{t},p_{t+1})$ of generation $G_{t}$, $t\geq0$, relative to first-period aggregate resources $\sum_{h\in G_{t}}\Vert e^{h}_{t}\Vert$ can only be achieved by a small real return rate $\Vert p_{t}\Vert/\Vert p_{t+1}\Vert$ relative to the first-period endowment growth rates $\sum_{h\in G_{t+1}}\Vert e^{h}_{t+1}\Vert/\sum_{h\in G_{t}}\Vert e^{h}_{t}\Vert$\footnote{The definition of prone-to-savings economies in \textcite[Definition 5, p. 6]{Dognini_2026a} adopted \textit{per capita} real savings instead of \textit{relative aggregate} real savings and compared the real return rates faced by each generation with the \textit{demographic} growth rate instead of the \textit{first-period endowment} growth rate.}. Stated differently, the level of aggregate real savings becomes an informative signal on the real return rate faced by each generation.

I proceed with the following assumption.

\begin{assumption}\label{assProneSavingsEconomy}
The economy $\mathcal{E}$ is prone to savings.
\end{assumption}

The lemma below reveals that in prone-to-savings economies, efficient equilibria must sustain, in all periods, the relative level of aggregate real savings above the threshold given by Definition \ref{defProneSavings}.

\begin{lemma}\label{lemmaSecondParetoOptEquilibria}
Under Assumptions \ref{assBoundsPopAndEndownments} and \ref{assBoundsPrices}--\ref{assProneSavingsEconomy}, if $p\in\mathcal{H}$ and $\sum_{h\in G_{t}}s^{h}(p_{t},p_{t+1})/\sum_{h\in G_{t}}\Vert e^{h}_{t}\Vert\leq \delta$, for some $t\geq0$,\footnote{To be fully clear, $\delta>0$ is given by Assumption \ref{assProneSavingsEconomy} and Definition \ref{defProneSavings}.} then $p\notin\mathcal{H}^{PO}$.
\end{lemma}

The next result gathers Lemma \ref{lemmaFirstParetoOptEquilibria} and Lemma \ref{lemmaSecondParetoOptEquilibria} to provide, through the dynamic of relative aggregate real savings, a complete characterization of the subset of Pareto optimal equilibria in prone-to-savings economies.

\begin{proposition}\label{propNecAndSuffConditionParetoOpt}
Let $p\in\mathcal{H}$. Under Assumptions \ref{assBoundsPopAndEndownments} and \ref{assBoundsPrices}--\ref{assProneSavingsEconomy}, $p\notin\mathcal{H}^{PO}$ if, and only if, 
\begin{eqnarray*}
    \lim_{t\rightarrow\infty}\frac{\sum_{h\in G_{t}}s^{h}(p_{t},p_{t+1})}{\sum_{h\in G_{t}} \Vert e^{h}_{t}\Vert}=0.
\end{eqnarray*}
\end{proposition}

Proposition \ref{propNecAndSuffConditionParetoOpt} reveals that inefficient equilibria are precisely those that depress the level of relative aggregate real savings. Therefore, it can be seen as an equivalent statement of the \textcite{Cass_1972} criterion that, instead of looking at the behavior of the series \ref{eqCass2}, deals only with the limit of its terms.

I proceed with the following assumption.

\begin{assumption}\label{assUniformBoundsDeltaAndL}
Let $L_{t}\in\mathbb{N},\sigma_{t}\in (0,1)$, $t\geq0$, be as defined above. Then, $\inf_{t\geq0}\sigma_{t}>0$ and $\sup_{t\geq0}L_{t}<\infty$
\end{assumption}

Assumption \ref{assUniformBoundsDeltaAndL} is a technical one that poses no significant constraint in the model. The next theorem is our first existence result.

\begin{theorem}\label{theoExistenceOfParetoOptJME}
Under Assumptions \ref{assBoundsPopAndEndownments}--\ref{assUniformBoundsDeltaAndL}, the subset of Pareto optimal equilibria $\mathcal{H}^{PO}\subseteq\mathcal{H}$ is non-empty and compact, with 
\begin{eqnarray}\label{eqMonetary}
    \frac{\sum_{h\in G_{0}} s^{h}(p_{0},p_{1})}{\sum_{h\in G_{0}}\Vert e^{h}_{0}\Vert}>\delta,
\end{eqnarray}
for $p\in\mathcal{H}^{PO}$.
\end{theorem}

Theorem \ref{theoExistenceOfParetoOptJME} provides sufficient conditions for $\mathcal{H}^{PO}$ to be non-empty and compact, the most relevant being that the economy is prone to savings. Furthermore, (\ref{eqMonetary}) implies that all the efficient equilibria in $\mathcal{H}^{PO}$ are, in a given sense, monetary\footnote{See a detailed discussion on this \textit{monetary} nature after Theorem 8 in \textcite[p. 7]{Dognini_2026a}.}.

The last result of this section provides conditions that ensure the existence of efficient monetary equilibria in classical overlapping generations economy with money (e.g., \textcite{CassOkunoZilcha_1979}, \textcite{BalaskoShell_1980,BalaskoShell_1981_1}, \textcite{OkunoZilcha_1980}, \textcite{KehoeLevine_1985} and \textcite{RaoAiyagari_1992}).

Let $\mathcal{E}^{2}$ be such an economy\footnote{I adopt the same notation as in \textcite[pp. 3-5]{Dognini_2026a}, where three different types of classical overlapping generations economies are detailed, and the ones with fiat money are of the second type (thus the superscript ``2'' in $\mathcal{E}^{2}$).}. The unique structural difference between $\mathcal{E}^{2}$ and $\mathcal{E}$ is that the younger period of life of the generation born in $t=0$ is not modeled. Therefore, household $h\in G^{2}_{0}$ is characterized by a continuous, non-decreasing, semi-strictly quasi-concave utility function $u^{h}:\mathbb{R}^{L_{1}}_{+}\rightarrow\mathbb{R}$ without local maxima, a nonzero endowment $e^{h}_{1}\in\mathbb{R}^{L_{1}}_{+}$ and an amount of fiat money $m^{h}\geq0$. I assume, without loss of generality, $\sum_{h\in G^{2}_{0}}m^{h}=1$. 

Also, the aggregate endowment satisfies $\sum_{t\geq0}\sum_{h\in G^{2}_{t}}e^{h}\in\mathbb{R}^{\infty}_{++}$. Furthermore, let $\mathcal{E}^{2}_{\geq1}$ be the economy formed by all generations of $\mathcal{E}^{2}$ born in $t\geq1$. Finally, the Walrasian demand of household $h\in G^{2}_{0}$, $x^{h}_{1}:\mathbb{R}_{+}\times\mathbb{R}^{L_{1}}_{++}\rightarrow\mathbb{R}^{L_{1}}_{+}$,  is now given by
\begin{eqnarray}\label{eqUMPG0}
    x^{h}_{1}(p_{m},p_{1})=\textrm{argmax}_{c\in\mathbb{R}^{L_{1}}_{+}}& u^{h} (c) \\
    \textrm{ s.t. }& p_{1}\cdot (c-e^{h}_{1})\leq p_{m}m^{h},\nonumber
\end{eqnarray}
for $(p_{m},p_{1})\in\mathbb{R}_{+}\times\mathbb{R}^{L_{1}}_{++}$. 

\begin{theorem}\label{theoClassModelExistence}
Let $\mathcal{E}^{2}$ be a classical economy with fiat money. Suppose $\mathcal{E}^{2}$ satisfies Assumptions \ref{assBoundsPopAndEndownments}, \ref{assResourceRelated} and \ref{assCassCriterion}; $\mathcal{E}^{2}_{\geq1}$ satisfies Assumptions $\ref{assBoundsPrices}$, \ref{assProneSavingsEconomy} and \ref{assUniformBoundsDeltaAndL}; and there is $\lambda>0$ such that\footnote{It is worth highlighting that (\ref{eqLowerBoundG20}) requires that, for prices $p_{1}\in\mathbb{R}^{L_{1}}_{++}$ stated in terms of the numeraire (i.e., $p_{m}=1$), if some commodity is too cheap, then the aggregate demand of the older generation becomes incompatible with Assumption \ref{assBoundsPopAndEndownments}.}
\begin{eqnarray}\label{eqLowerBoundG20}
    \min_{1\leq i\leq L_{1}}\, p_{1i}<\lambda\implies  \sum_{h\in G^{2}_{0}}\Vert x^{h}_{1}(1,p_{1})\Vert>\beta\sum_{h\in G^{2}_{0}}\Vert e^{h}_{1}\Vert,
\end{eqnarray}
for $p_{1}\in\mathbb{R}^{L_{1}}_{++}$. Then, there exists a Pareto optimal monetary equilibrium of $\mathcal{E}^{2}$.
\end{theorem}

Theorem \ref{theoClassModelExistence} provides sufficient conditions for the existence of efficient monetary equilibria in classical non-stationary consumption-loan overlapping generations economies with fiat money. In short, it ``states that if the economy is prone to savings, then money has value and leads the economy to an efficient allocation, thus furnishing a solution to the \textcite{Hahn_1965} problem'' \parencite[p. 9]{Dognini_2026a}. 

\section{A Taylor rule based on real aggregate savings}\label{sec3}
I consider in this section a two-period overlapping generations economy with fiat money that differs from the classical one $\mathcal{E}^{2}$ in one aspect: Government now exists and, through monetary policy, influences the price level and long-term (i.e., interperiod) savings. Let $\mathcal{E}^{G}$ be this economy (the ``G'' stands for ``Government''). 

The real fundamentals (i.e., preferences and endowments, since we are dealing with consumption-loan economies) of $\mathcal{E}^{G}$ are the same as those of $\mathcal{E}^{2}$ and all assumptions from Theorem \ref{theoClassModelExistence} are satisfied, so that the existence of monetary equilibria is assured.

In the beginning of period $t\geq1$, Government has an outstanding amount of nominal debt $\sum_{h\in G_{t-1}}b^{h}\geq0$ sold to the generation that was born at $t-1\geq0$. Government pays these debts with fiat money, giving  $m^{h}=b^{h}$ to household $h\in G_{t-1}$. It also sells new bonds to the younger generation at price $(1+r_{t})^{-1}>0$, where $r_{t}\geq0$ is the nominal interperiod interest rate fixed through monetary policy in period $t\geq1$. I assume that $r_{t}\geq0$ since, otherwise, young households would prefer to hold money ``below their mattresses''. 

In each period, the Government budget balances so that
\begin{eqnarray}\label{eqGovBudget}
    \sum_{h\in G_{t-1}}b^{h}=\frac{1}{1+r_{t}}\sum_{h\in G_{t}}b^{h},
\end{eqnarray}
for $t\geq1$. Given the nominal interest rate $r_{t}\geq0$ and spot prices $(q_{t},q_{t+1})\in\mathbb{R}^{L_{t}+L_{t+1}}_{++}$ stated in terms of current fiat money, household $h\in G_{t}$, $t\geq1$, solves
\begin{eqnarray*}
    \tilde{x}^{h}(r_{t},q_{t},q_{t+1})=\textrm{argmax}_{(m_{t},c)\in\mathbb{R}\times\mathbb{R}^{L_{t}+L_{t+1}}_{+}}& u^{h} (c) \\
    \textrm{ s.t. }& q_{t}\cdot c_{t}+m_{t}=q_{t}\cdot e^{h}_{t}\\
    & (1+r_{t})^{-1}b=m_{t}\\
    &m_{t+1}=b\\
    &q_{t+1}\cdot c_{t+1}=q_{t+1}\cdot e^{h}_{t+1}+m_{t+1}.
\end{eqnarray*}
The optimization problem is written this way to illustrate the following sequence of actions. When young, households may sell their goods for fiat money $m_{t}>0$\footnote{Although I describe the case with $m_{t}>0$, we could have households $h\in G_{t}$ with $m_{t}=0$ (i.e., households that do not want to buy any bonds) or even $m_{t}<0$, what would change the interpretation of the model. This last case represents households that are borrowers from their own cohort, and one can see a detailed discussion on this matter after Theorem 8 in \textcite{Dognini_2026a}.}. Then, they can use this fiat money to buy $b=m_{t}(1+r_{t})>0$ bonds and finance excess old-age consumption. When old, their bonds are paid by the Government, yielding $m_{t+1}=b$, and with this fiat money in hand, they consume more than their endowment could afford. 

Clearly, this optimization can also be written briefly as
\begin{eqnarray}\label{eqUMP2}
    \tilde{x}^{h}(r_{t},q_{t},q_{t+1})=\textrm{argmax}_{c\in\mathbb{R}^{L_{t}+L_{t+1}}_{+}}& u^{h} (c) \\
    \textrm{ s.t. }& (q_{t},(1+r_{t})^{-1}q_{t+1})\cdot (c-e^{h})=0.\nonumber
\end{eqnarray}
Therefore, the real return rate on savings $r^{*}_{t}>-1$ satisfies the following Fisher relation
\begin{eqnarray}\label{eqFisher}
    1+r^{*}_{t}=\frac{\Vert q_{t}\Vert}{\Vert q_{t+1}\Vert}(1+r_{t})=\frac{1+r_{t}}{1+\pi_{t}},
\end{eqnarray}
with $\pi_{t}=\Vert q_{t+1}\Vert/\Vert q_{t}\Vert-1$, $t\geq1$, measuring inflation during generation $G_{t}$ lifespan. Notice that our reference baskets for measuring the real return rates $r^{*}_{t}$ and inflation rates $\pi_{t}$, $t\geq1$, are $(1,\ldots,1)\in\mathbb{R}^{L_{t}}_{++}$, $t\geq1$.

Furthermore, comparing (\ref{eqUMP}) and (\ref{eqUMP2}), we can see that for $h\in G_{t}$, $t\geq1$,
\begin{eqnarray}\label{eqEquivalencePrices}
    \tilde{x}^{h}(r_{t},q_{t},q_{t+1})=x^{h}(q_{t},(1+r_{t})^{-1}q_{t+1}),
\end{eqnarray}
with $r_{t}\geq0$, $(q_{t},q_{t+1})\in\mathbb{R}^{L_{t}+L_{t+1}}_{++}$ (notice that the $\sim$ overscript differentiates the forms in which the arguments of the Walrasian demands are posed).

Household $h\in G_{0}$ solves the following optimization problem
\begin{eqnarray}\label{eqUMPG02}
    \tilde{x}^{h}(q_{1})=\textrm{argmax}_{c\in\mathbb{R}^{L_{1}}_{+}}& u^{h} (c) \\
    \textrm{ s.t. }& m_{1}=b^{h}\nonumber\\
    &q_{1}\cdot c_{t}=m_{1}+q_{1}\cdot e^{h}_{1},\nonumber
\end{eqnarray}
with $b^{h}\geq0$ the bonds bought in the previous period (i.e., $t=0$). Comparing (\ref{eqUMPG0}) and (\ref{eqUMPG02}), we can see that for $h\in G_{0}$,
\begin{eqnarray}\label{eqEquivalencePricesG0}
    \tilde{x}^{h}(q_{1})=x^{h}(1,q_{1}),
\end{eqnarray}
with $q_{1}\in\mathbb{R}^{L_{1}}_{++}$. In the following definition, $r=(r_{1},r_{2},\ldots)\in\mathbb{R}^{\infty}_{+}$ is a \textit{monetary policy} (i.e., a sequence of nominal interest rates), and $q=(q_{1},q_{2},\ldots)\in\mathbb{R}^{\infty}_{++}$ is a sequence of \textit{spot prices}. 

\begin{definition}\label{defEquilibriaGovernment}
    Let $\mathcal{E}^{G}$ be the economy described above. Then, $(r,q)\in\mathbb{R}^{\infty}_{+}\times\mathbb{R}^{\infty}_{++}$ is an \textit{equilibrium} of $\mathcal{E}^{G}$ if
    \begin{eqnarray*}
    \sum_{h\in G_{0}}\tilde{x}^{h}(q_{1})+\sum_{h\in G_{1}}\tilde{x}^{h}(r_{1},q_{1},q_{2})=\sum_{h\in G_{0}}e^{h}_{1}+\sum_{h\in G_{1}}e^{h}_{1},
    \end{eqnarray*}
    and
    \begin{eqnarray*}
    \sum_{h\in G_{t-1}}\tilde{x}^{h}(r_{t-1},q_{t-1},q_{t})+\sum_{h\in G_{t}}\tilde{x}^{h}(r_{t},q_{t},q_{t+1})=\sum_{h\in G_{t-1}}e^{h}_{t}+\sum_{h\in G_{t}}e^{h}_{t},
    \end{eqnarray*}
    for $t\geq2$. Furthermore, the \textit{set of equilibria} of $\mathcal{E}^{G}$ is denoted by $\mathcal{H}^{G}\subseteq \mathbb{R}^{\infty}_{+}\times\mathbb{R}^{\infty}_{++}$.
\end{definition}
An equilibrium of $\mathcal{E}^{G}$, therefore, is a stated monetary policy $r\in\mathbb{R}^{\infty}_{+}$ and a sequence of spot prices $q\in\mathbb{R}^{\infty}_{++}$ that clear all goods markets and, due to Walras' law, also satisfies (\ref{eqGovBudget}). Given $(r,q)\in\mathcal{H}^{G}$, let $p_{1}=q_{1}$ and
\begin{eqnarray}\label{eqPQEquivalence}
    p_{t}=\frac{q_{t}}{\prod^{t-1}_{i=1}(1+r_{i})}
\end{eqnarray}
for $t\geq2$. Then,
\begin{eqnarray*}
    \sum_{h\in G_{0}}x^{h}(1,p_{1})+\sum_{h\in G_{1}}x^{h}(p_{1},p_{2})&=&\sum_{h\in G_{0}}x^{h}(1,q_{1})+\sum_{h\in G_{1}}x^{h}(q_{1},(1+r_{1})^{-1}q_{2})\\
    &=&\sum_{h\in G_{0}}\tilde{x}^{h}(q_{1})+\sum_{h\in G_{1}}\tilde{x}^{h}(r_{1},q_{1},q_{2})\\
    &=&\sum_{h\in G_{0}}e^{h}_{1}+\sum_{h\in G_{1}}e^{h}_{1}.
\end{eqnarray*}
where the penultimate equality is due to (\ref{eqEquivalencePrices}) and (\ref{eqEquivalencePricesG0}), and the last to Definition \ref{defEquilibriaGovernment}. Also, 
\begin{eqnarray*}
    \sum_{h\in G_{t-1}}x^{h}(p_{t-1},p_{t})+\sum_{h\in G_{t}}x^{h}(p_{t},p_{t+1})&=&\sum_{h\in G_{t-1}}x^{h}(q_{t-1},(1+r_{t-1})^{-1}q_{t})+\sum_{h\in G_{t}}x^{h}(q_{t},(1+r_{t})^{-1}q_{t+1})\\
    &=&\sum_{h\in G_{t-1}}\tilde{x}^{h}(r_{t-1},q_{t-1},q_{t})+\sum_{h\in G_{t}}\tilde{x}^{h}(r_{t},q_{t},q_{t+1})\\
    &=&\sum_{h\in G_{t-1}}e^{h}_{t}+\sum_{h\in G_{t}}e^{h}_{t},
\end{eqnarray*}
for $t\geq2$, where the first equality is due to homogeneity of Walrasian demand, the second to (\ref{eqEquivalencePrices}) and (\ref{eqEquivalencePricesG0}), and the last to Definition \ref{defEquilibriaGovernment}.
To ease notation, let $\rho:\mathcal{H}^{G}\rightarrow\mathbb{R}^{\infty}_{++}$, $\rho(r,q)=p$. Then, $\textbf{Img\,}\rho\subseteq \mathcal{H}^{2}$, where $\mathcal{H}^{2}$ is the set of monetary equilibria of $\mathcal{E}^{2}$ (i.e., every equilibrium of $\mathcal{E}^{G}$ can be associated with a monetary equilibrium of $\mathcal{E}^{2}$ that completely characterizes the real allocations of the economy).

We do not know yet, however, if $\mathcal{H}^{G}\neq\emptyset$. Theorem \ref{theoClassModelExistence} implies that $\mathcal{H}^{2}\neq\emptyset$. Let $p\in\mathcal{H}^{2}$ and $r\in\mathbb{R}^{\infty}_{+}$ be any monetary policy. Define $q_{1}=p_{1}$ and
\begin{eqnarray}\label{eqPQ2}
    q_{t}=p_{t}\prod^{t-1}_{i=1}(1+r_{i})
\end{eqnarray}
for $t\geq2$. Then,
\begin{eqnarray*}
    \sum_{h\in G_{0}}e^{h}_{1}+\sum_{h\in G_{1}}e^{h}_{1}&=&\sum_{h\in G_{0}}x^{h}(1,p_{1})+\sum_{h\in G_{1}}x^{h}(p_{1},p_{2})\\
    &=&\sum_{h\in G_{0}}x^{h}(1,q_{1})+\sum_{h\in G_{1}}x^{h}(q_{1},(1+r_{1})^{-1}q_{2})\\
    &=&\sum_{h\in G_{0}}\tilde{x}^{h}(q_{1})+\sum_{h\in G_{1}}\tilde{x}^{h}(r_{1},q_{1},q_{2})
\end{eqnarray*}
where the first equality is due to the definition of $\mathcal{H}^{2}$, the second to (\ref{eqPQ2}), and the last to (\ref{eqEquivalencePrices}) and (\ref{eqEquivalencePricesG0}). Also, 
\begin{eqnarray*}
    \sum_{h\in G_{t-1}}e^{h}_{t}+\sum_{h\in G_{t}}e^{h}_{t}&=&\sum_{h\in G_{t-1}}x^{h}(p_{t-1},p_{t})+\sum_{h\in G_{t}}x^{h}(p_{t},p_{t+1})\\
    &=&\sum_{h\in G_{t-1}}x^{h}(q_{t-1},(1+r_{t-1})^{-1}q_{t})+\sum_{h\in G_{t}}x^{h}(q_{t},(1+r_{t})^{-1}q_{t+1})\\
    &=&\sum_{h\in G_{t-1}}\tilde{x}^{h}(r_{t-1},q_{t-1},q_{t})+\sum_{h\in G_{t}}\tilde{x}^{h}(r_{t},q_{t},q_{t+1})
\end{eqnarray*}
for $t\geq2$, where the first equality is due to the definition of $\mathcal{H}^{2}$, the second to (\ref{eqPQ2}) and to the homogeneity of Walrasian demand, and the last to (\ref{eqEquivalencePrices}) and (\ref{eqEquivalencePricesG0}). Therefore, $(r,q)\in\mathcal{H}^{G}$ and so $\mathcal{H}^{G}\neq\emptyset$. 

Actually, this reasoning shows a stronger statement: that for all monetary policy $r\in\mathbb{R}^{\infty}_{+}$, $\textbf{Img\,}\rho(r,\cdot)=\mathcal{H}^{2}$, i.e., a \textit{fixed} monetary policy, whatever it may be, does not impose any constraints on the possible real allocations of the economy (a result that is reminiscent of \textcite{1975_SargentWallace}).

Next, suppose that $(r,q),(r^{\prime},q^{\prime})\in\mathcal{H}^{G}$ are two equilibria associated with the same monetary equilibrium $p\in\mathcal{H}^{2}$ (i.e., $\rho(r,q)=\rho(r^{\prime},q^{\prime})$). Then, (\ref{eqPQEquivalence}) implies $q_{1}=q^{\prime}_{1}$ and
\begin{eqnarray*}
    \frac{ q_{t}}{\prod^{t-1}_{i=1}(1+r_{i})}=\frac{q^{\prime}_{t}}{\prod^{t-1}_{i=1}(1+r^{\prime}_{i})}\implies \frac{\Vert q_{t}\Vert}{\Vert q^{\prime}_{t}\Vert}=\prod^{t-1}_{i=1}\frac{1+r_{i}}{1+r^{\prime}_{i}},
\end{eqnarray*}
for $t\geq2$. Therefore,
\begin{eqnarray*}
    \frac{1+\pi_{t}}{1+\pi^{\prime}_{t}}=\frac{\Vert q_{t+1}\Vert \Vert q^{\prime}_{t}\Vert}{\Vert q_{t}\Vert \Vert q^{\prime}_{t+1}\Vert}=\frac{1+r_{t}}{1+r^{\prime}_{t}}\implies r^{*}_{t}=r^{*\prime}_{t},
\end{eqnarray*}
for $t\geq1$, so that two monetary policy and spot prices define the same real outcomes only if real return rates are kept unchanged along with relative prices in each period. All these results are gathered in the next proposition.
\begin{proposition}\label{propCardinalitySetsEquilibria}
    Let $\mathcal{H}^{2}$ and $\rho(\cdot)$ be as defined above. Then, $\textbf{Img\,}\rho(r,\cdot)=\mathcal{H}^{2}$, $r\in\mathbb{R}^{\infty}_{+}$, and if $\rho(r,q)=\rho(r^{\prime},q^{\prime})$, then $q_{1}=q^{\prime}_{1}$, $q_{t}/\Vert q_{t}\Vert=q^{\prime}_{t}/\Vert q^{\prime}_{t}\Vert$, $t\geq2$, and $r^{*}_{t}=r^{*\prime}_{t}$, $t\geq1$.
\end{proposition}

Proposition \ref{propCardinalitySetsEquilibria} reveals that $\#\mathcal{H}^{G}>\#\mathcal{H}^{2}$ (i.e., the cardinality of $\mathcal{H}^{G}$ is greater than that of $\mathcal{H}^{2}$) and, therefore, the multiplicity of equilibria that is inherent to $\mathcal{E}^{2}$ is reinforced when we deal with $\mathcal{E}^{G}$. In this sense, monetary policy increases the degree of indeterminacy of the economy.

To tackle this matter of indeterminacy, I proceed with the following definition.

\begin{definition}\label{defTaylorRule}
    A \textit{Taylor rule} is a function $\mathcal{T}:\mathcal{H}^{G}\rightarrow \mathbb{R}^{\infty}_{+}$ and an equilibrium $(r,q)\in \mathcal{H}^{G}$ is \textit{stable} relative to $\mathcal{T}(\cdot)$ if $r=\mathcal{T}(r,q)$. Furthermore, the \textit{set of stable equilibria} relative to $\mathcal{T}(\cdot)$ is denoted by $\mathcal{S}_{\mathcal{T}}\subseteq\mathcal{H}^{G}$.
\end{definition}

Definition \ref{defTaylorRule} provides a rule $\mathcal{T}(\cdot)$ under which the Government updates a given announced monetary policy $r\in\mathbb{R}^{\infty}_{+}$ based on this policy and the sequence of equilibrium spot prices $q\in\mathbb{R}^{\infty}_{++}$ that the economy has ``landed in''. This rule, therefore, drives the interaction between the Government and the overall economy (i.e., households) in the following way. 

Imagine that Government announces the monetary policy $r^{1}\in\mathbb{R}^{\infty}_{+}$ and the economy moves towards an equilibrium $(r^{1},q^{1})\in\mathcal{H}^{G}$. If the equilibrium outcome is compatible with the Taylor rule (i.e., if $r^{1}=\mathcal{T}(r^{1},q^{1})$), then no further action is taken. If not, the Government uses the rule to update his monetary policy to $r^{2}=\mathcal{T}(r^{1},q^{1})$, and the economy moves towards a new equilibrium $(r^{1},q^{1})\in\mathcal{H}^{G}$, in a sort of \textit{tâtonnement process}. 

Notice that I am \textit{not} modeling this tâtonnement dynamic; the matter of \textit{how} and to \textit{which} monetary equilibrium the economy will move is an engine that lies outside of this model. I do, however, assume that this process will eventually stabilize, and I am interested in finding the equilibria that are \textit{stable} relative to the Taylor rule $\mathcal{T}(\cdot)$. 

Let us first analyze the Taylor rule given by
\begin{eqnarray}\label{eqFirstTaylor}
    \mathcal{T}_{1}(r,q)=(r_{1}+\theta(\pi_{0}-\pi^{*}),r_{2}+\theta(\pi_{1}-\pi^{*}),\ldots)
\end{eqnarray}
with $\theta>0$ the response coefficient, $\Vert q_{0}\Vert>0$ given (i.e., the ``last year'' price level), and $\pi^{*}>-1$ the \textit{inflation target}. The Taylor rule $\mathcal{T}_{1}(\cdot)$ increases nominal rates when inflation is greater than the target and decreases nominal rates when inflation is lower than the target. Notice that if $(r,q)\in\mathcal{H}^{G}$ is stable relative to $\mathcal{T}_{1}(\cdot)$, then $\pi_{t}=\pi^{*}$, $t\geq0$. Therefore, $\mathcal{T}_{1}(\cdot)$ does bring the inflation rate towards its target. Furthermore, for $p=\rho(r,q)$, $(r,q)\in\mathcal{S}_{\mathcal{T}_{1}}$, we have $p\in\mathcal{H}^{2}$ and $\Vert p_{1}\Vert=\Vert q_{1}\Vert=(1+\pi^{*})\Vert q_{0}\Vert$, so that
\begin{eqnarray*}
    \rho(\mathcal{S}_{\mathcal{T}_{1}})\subseteq\{p\in\mathcal{H}^{2}\mid \Vert p_{1}\Vert=(1+\pi^{*})\Vert q_{0}\Vert\},
\end{eqnarray*}
and, therefore,
\begin{eqnarray}\label{eqInclusionT1}
    \mathcal{S}_{\mathcal{T}_{1}}\subseteq \rho^{-1}(\{p\in\mathcal{H}^{2}\mid \Vert p_{1}\Vert=(1+\pi^{*})\Vert q_{0}\Vert\})\cap\{(r,q)\in\mathcal{H}^{G}\mid \pi_{t}=\pi^{*}, t\geq0\}
\end{eqnarray}
Next, suppose there is $p\in\mathcal{H}^{2}$ with $\Vert p_{1}\Vert=(1+\pi^{*})\Vert q_{0}\Vert$, and define
\begin{eqnarray}
    q^{\prime}_{t}&=&\frac{p_{t}}{\Vert p_{t}\Vert}\Vert q_{0}\Vert(1+\pi^{*})^{t}\label{eqAux1}\\
    1+r^{\prime}_{t}&=&(1+\pi^{*})\frac{\Vert p_{t}\Vert}{\Vert p_{t+1}\Vert}\label{eqAux2},
\end{eqnarray}
for $t\geq1$. Notice that
\begin{eqnarray*}
    q^{\prime}_{1}=\frac{p_{1}}{\Vert p_{1}\Vert}\Vert q_{0}\Vert(1+\pi^{*})\implies q^{\prime}_{1}=p_{1},
\end{eqnarray*}
and that
\begin{eqnarray*}
    \frac{q^{\prime}_{t}}{\prod^{t-1}_{i=1}(1+r^{\prime}_{i})}=\frac{q^{\prime}_{t}}{\prod^{t-1}_{i=1}(1+\pi^{*})\Vert p_{i}\Vert/\Vert p_{i+1}\Vert}=\frac{q^{\prime}_{t}\Vert p_{t}\Vert}{(1+\pi^{*})^{t-1}\Vert p_{1}\Vert}=\frac{q^{\prime}_{t}\Vert p_{t}\Vert}{(1+\pi^{*})^{t}\Vert q_{0}\Vert}=p_{t},
\end{eqnarray*}
for $t\geq2$, where the first equality is due to (\ref{eqAux2}) and the third to (\ref{eqAux1}). We conclude that $\rho(r,q)=p$. Then, (\ref{eqAux1}) allows us to write  
\begin{eqnarray*}
    1+\pi_{t}=\frac{\Vert q^{\prime}_{t+1}\Vert}{\Vert q^{\prime}_{t}\Vert}=1+\pi^{*}\implies \pi_{t}=\pi^{*},
\end{eqnarray*}
for $t\geq0$, and we have
\begin{eqnarray*}
    (r^{\prime},q^{\prime})\in \rho^{-1}(\{p\in\mathcal{H}^{2}\mid \Vert p_{1}\Vert=(1+\pi^{*})\Vert q_{0}\Vert\})\cap\{(r,q)\in\mathcal{H}^{G}\mid \pi_{t}=\pi^{*}, t\geq0\}.
\end{eqnarray*}
Suppose $(\tilde{r},\tilde{q})\in \rho^{-1}(\{p\in\mathcal{H}^{2}\mid \Vert p_{1}\Vert=(1+\pi^{*})\Vert q_{0}\Vert\})\cap\{(r,q)\in\mathcal{H}^{G}\mid \pi_{t}=\pi^{*}, t\geq0\}$. Then, $\rho(\tilde{r},\tilde{q})=p=\rho(r^{\prime},q^{\prime})$, and Proposition \ref{propCardinalitySetsEquilibria} implies: $\tilde{q}_{1}=q^{\prime}_{1}=p_{1}\Vert q_{0}\Vert(1+\pi^{*})/\Vert p_{1}\Vert$; $\tilde{q}_{t}/\Vert \tilde{q}_{t}\Vert=q^{\prime}_{t}/\Vert q^{\prime}_{t}\Vert$, $t\geq2$; and $r^{*}_{t}=r^{*\prime}_{t}$, $t\geq1$. From this last equality of real return rates we obtain trhough the Fisher relation
\begin{eqnarray*}
    \frac{1+r^{\prime}_{t}}{1+\pi^{\prime}_{t}}=\frac{1+\tilde{r}_{t}}{1+\tilde{\pi}_{t}}\implies r^{\prime}_{t}=\tilde{r}_{t},
\end{eqnarray*}
for $t\geq1$, where the implication is due to the fact that $\pi_{t}^{\prime}=\tilde{r}_{t}=\pi^{*}$, $t\geq1$. Since all inflation rates are equal, then $\Vert q^{\prime}_{t}\Vert=\Vert \tilde{q}_{t}\Vert=\Vert q_{0}\Vert(1+\pi^{*})^{t}$, $t\geq1$. We conclude that $(\tilde{r},\tilde{q})=(r^{\prime},q^{\prime})$ and, therefore,
\begin{eqnarray}\label{eqUniqueElement}
    \{(r^{\prime},q^{\prime})\}=\rho^{-1}(\{p\})\cap\{(r,q)\in\mathcal{H}^{G}\mid \pi_{t}=\pi^{*}, t\geq0\}.
\end{eqnarray}
Notice again that $\Vert p_{1}\Vert=(1+\pi^{*})\Vert q_{0}\Vert$ implies $p_{1}=q^{\prime}_{1}$ and allows us to write
\begin{eqnarray*}
    \sum_{h\in G_{0}}e^{h}_{1}+\sum_{h\in G_{1}}e^{h}_{1}&=&\sum_{h\in G_{0}}x^{h}(1,p_{1})+\sum_{h\in G_{1}}x^{h}(p_{1},p_{2})\\
    &=&\sum_{h\in G_{0}}x^{h}(1,q^{\prime}_{1})+\sum_{h\in G_{1}}x^{h}(q^{\prime}_{1},(1+r^{\prime}_{1})^{-1}q^{\prime}_{2})\\
    &=&\sum_{h\in G_{0}}\tilde{x}^{h}(q^{\prime}_{1})+\sum_{h\in G_{1}}\tilde{x}^{h}(r^{\prime}_{1},q^{\prime}_{1},q^{\prime}_{2}),
\end{eqnarray*}
where the first equality is due to the definition of $\mathcal{H}^{2}$, the second to (\ref{eqAux1}) and (\ref{eqAux2}), and the third to (\ref{eqEquivalencePrices}) and (\ref{eqEquivalencePricesG0}). Also, 
\begin{eqnarray*}
    \sum_{h\in G_{t-1}}e^{h}_{t}+\sum_{h\in G_{t}}e^{h}_{t}&=&\sum_{h\in G_{t-1}}x^{h}(p_{t-1},p_{t})+\sum_{h\in G_{t}}x^{h}(p_{t},p_{t+1})\\
    &=&\sum_{h\in G_{t-1}}x^{h}(q^{\prime}_{t-1},(1+r^{\prime}_{t-1})^{-1}q^{\prime}_{t})+\sum_{h\in G_{t}}x^{h}(q^{\prime}_{t},(1+r^{\prime}_{t})^{-1}q^{\prime}_{t+1})\\
    &=&\sum_{h\in G_{t-1}}\tilde{x}^{h}(r^{\prime}_{t-1},q^{\prime}_{t-1},q^{\prime}_{t})+\sum_{h\in G_{t}}\tilde{x}^{h}(r^{\prime}_{t},q^{\prime}_{t},q^{\prime}_{t+1}),
\end{eqnarray*}
for $t\geq2$, where the first equality is due to the definition of $\mathcal{H}^{2}$, the second to (\ref{eqAux1}), (\ref{eqAux2}), and the homogeneity of Walrasian demand, and the last to (\ref{eqEquivalencePrices}) and (\ref{eqEquivalencePricesG0}). Therefore, $(r^{\prime},q^{\prime})\in\mathcal{S}_{\mathcal{T}_{1}}$ and, due to  (\ref{eqUniqueElement}), we have
\begin{eqnarray}\label{eq2ndInclusionT1}
    \rho^{-1}(\{p\in\mathcal{H}^{2}\mid \Vert p_{1}\Vert=(1+\pi^{*})\Vert q_{0}\Vert\})\cap\{(r,q)\in\mathcal{H}^{G}\mid \pi_{t}=\pi^{*}, t\geq0\}\subseteq\mathcal{S}_{\mathcal{T}_{1}}.
\end{eqnarray}
Finally, (\ref{eqInclusionT1}) and (\ref{eq2ndInclusionT1}) are gathered in the following proposition.
\begin{proposition}\label{propT1Stable}
    Let $\mathcal{H}^{2}$, $\rho(\cdot)$, $\mathcal{T}_{1}(\cdot)$ and $\mathcal{S}_{\mathcal{T}_{1}}$ be as defined above. Then, 
    \begin{eqnarray*}
        \mathcal{S}_{\mathcal{T}_{1}}=\rho^{-1}(\{p\in\mathcal{H}^{2}\mid \Vert p_{1}\Vert=(1+\pi^{*})\Vert q_{0}\Vert\})\cap\{(r,q)\in\mathcal{H}^{G}\mid \pi_{t}=\pi^{*}, t\geq0\}.
    \end{eqnarray*}
\end{proposition}

Proposition \ref{propT1Stable} identifies the set of stable equilibria of the Taylor rule $\mathcal{T}_{1}$ and reveals, in particular, that it does not depend on the actual value of $\theta>0$ (i.e., if the monetary policy is active, $\theta>1$, or not, $\theta\in(0,1]$).

Also, this result poses two concerns. First, if the target $\pi^{*}>-1$ is too low, $\mathcal{S}_{\mathcal{T}_{1}}$ may become empty and, therefore, no stable equilibria relative to $\mathcal{T}_{1}$ may exist. Second, even if $\mathcal{S}_{\mathcal{T}_{1}}\neq\emptyset$, there is no guaranty that $\rho(\mathcal{S}_{\mathcal{T}_{1}})\subseteq\mathcal{H}^{2PO}$ (i.e., we are not assured that the Taylor rule will select an \textit{efficient} monetary equilibrium). The next example illustrates how the Taylor rule $\mathcal{T}_{1}(\cdot)$ can lead to the nonexistence of stable equilibria or the selection of inefficient ones.
\begin{example}\label{ex1}
    Let $L_{t}=1$, $t\geq1$. There is a single household in each generation $G_{t}$, $t\geq1$, a common utility function $u(c_{1},c_{2})=\sqrt{c_{1}}+\sqrt{c_{2}}$, and endowments are given by $e^{1}=(1,0)$, $e^{2}=(d,0)$, $e^{t}=(d(1+g)^{t-2},0)$, $t\geq3$, $g>0$ (i.e., $d>0$ represents a possible productivity shock, after which the economy evolves at a constant rate $1+g>0$). The single household of generation $G_{0}$ has $b^{0}=1$ bonds and no endowment.

    The equilibrium equations for this economy can be written as
    \begin{eqnarray*}
        1-x^{1}_{1}(p_{1},p_{2})&=&\frac{1}{p_{1}}\\
        d-x^{2}_{1}(p_{2},p_{3})&=&x^{1}_{2}(p_{1},p_{2})\\
        d\eta-x^{3}_{1}(p_{3},p_{4})&=&x^{2}_{2}(p_{2},p_{3})\\
        d\eta^{2}-x^{4}_{1}(p_{4},p_{5})&=&x^{3}_{2}(p_{3},p_{4})\\
        &\vdots&
    \end{eqnarray*}

    Solving the system backwards (see \textcite[Section 4, pp. 8-10]{Dognini_2026a} for a discussion on backward calculation algorithms), we know that the monetary equilibrium price sequences in $\mathcal{H}^{2}$ are uniquely defined by the condition $p_{3}=p_{2}/\eta$, $0<\eta\leq (1+g)$. Then, the second equation implies
    \begin{eqnarray*}
        d-x^{2}_{1}(p_{2},p_{2}/\eta)=\frac{\eta d}{1+\eta}=\frac{p_{1}}{p_{2}}\frac{p_{1}}{p_{1}+p_{2}},
    \end{eqnarray*}
    for $0<\eta\leq (1+g)$, and the first equation implies
    \begin{eqnarray*}
        1-x^{1}_{1}(p_{1},p_{2})=\frac{p_{1}}{p_{1}+p_{2}}=\frac{1}{p_{1}},
    \end{eqnarray*}
    so that $p_{2}=(1+\eta)/\eta d$ and, therefore,
    \begin{eqnarray*}
        p_{1}=\frac{1+\sqrt{1+4(1+\eta)/\eta d}}{2}=f(d,\eta),
    \end{eqnarray*}
    for $0<\eta\leq (1+g)$. Next, assume, without loss of generality, that the price level at the immediate past was $q_{0}=1$ and that there is $(r,q)\in\mathcal{S}_{\mathcal{T}_{1}}$. Then,
    \begin{eqnarray}\label{eqInflationTarget}
        1+\pi^{*}=1+\pi_{0}=\frac{q_{1}}{q_{0}}=p_{1}=f(d,\eta),
    \end{eqnarray}
    for $0<\eta\leq(1+g)$. In particular, we have $\pi^{*}\geq f(d,1+g)-1>0$. Therefore, if $\pi^{*}<f(d,1+g)-1$, then $\mathcal{S}_{\mathcal{T}_{1}}=\emptyset$ (i.e., an inflation target that is too low disrupts all possibilities of stable equilibria). 
    
    Notice that the threshold $f(d,1+g)$ is strictly decreasing in both its arguments. Therefore, the magnitude of negative shocks that can hit the economy (represented here by low values of $d>0$) define how low the inflation target can go before disrupting stable equilibria (i.e., larger shocks lead to higher targets). Also, the prospects of future growth rates (represented here by $1+g>0$) also define how low the inflation target can go. If the prospect is of high growth, than the inflation target can go low; if not, it must accommodate the necessary real adjustments of the economy.

    Next, we can use (\ref{eqInflationTarget}) to find $\eta\in(0,1+g]$ (which defines the equilibria in $\mathcal{H}^{2}$) as a function of $\pi^{*}$, thus leading us to
    \begin{eqnarray*}
        \eta=\frac{1}{\pi^{*}(1+\pi^{*})d-1}.
    \end{eqnarray*}
    In particular, there is $\pi_{g}>0$ such that $(\pi_{g}(1+\pi_{g})d-1)^{-1}=1+g$. Notice that all values of $\pi^{*}\in[\pi_{g},+\infty)$ lead to $\mathcal{S}_{\mathcal{T}_{1}}\neq\emptyset$ (i.e., there is a continuum of possible inflation targets that do not disrupt the stable equilibria of the economy). 
    
    Furthermore, each of these possible values of $\pi^{*}$ selects a single monetary equilibrium $p\in\mathcal{H}^{2}$. However, there is a single inflation target that does selects the unique \textit{efficient monetary equilibrium}, which is precisely $\pi^{*}=\pi_{g}$. In all other cases (i.e., $\pi^{*}>\pi_{g}$), the Taylor rule $\mathcal{T}_{1}$ will lead the economy to an inefficient equilibrium. 
\end{example}

Example \ref{ex1} reveals that the Taylor rule (\ref{eqFirstTaylor}) has three properties: (i) it brings inflation towards its target $\pi^{*}>0$ in all stable equilibria; (ii) if the target is set too low, the rule leads to the non-existence of stable equilibria, and the threshold is related to the magnitude of adverse economic shocks and future growth rate prospects (i.e., the target must allow inflation to be used as an adjustment tool to structural worsening of the economic landscape); (iii) even if the target does not disrupt stable equilibria (i.e., if $\mathcal{S}_{\mathcal{T}_{1}}\neq\emptyset$), there is no guaranty that $\mathcal{T}_{1}$ will select an \textit{efficient monetary equilibrium}.

Actually, Example \ref{ex1} shows that if the target is set as $\pi^{*}>\pi_{g}$, then we \textit{certainly} will not select an efficient monetary equilibrium. Therefore, the target must be set with pinpoint accuracy: if $\pi^{*}>\pi^{g}$, the economy will be led to an inefficient monetary equilibrium; if $\pi^{*}<\pi^{g}$, the target will disrupt the stability of equilibria.

In order to forego this need for razor-sharp precision, let us define the following Taylor rule\footnote{For $A\subseteq \mathbb{R}$, the \textit{indicator function} $\mathbbm{1}_{A}:\mathbb{R}\rightarrow\{0,1\}$ is given by $\mathbbm{1}_{A}(x)=1$ if, and only if, $x\in A$.}
\begin{eqnarray}\label{eqSecondTaylor}
    \mathcal{T}_{2}(r,q)=(r_{1}+\theta\mathbbm{1}_{(\pi^{*},\infty)}(\pi_{0}),r_{2}+\mathbbm{1}_{(\pi^{*},\infty)}(\pi_{1}),\ldots)
\end{eqnarray}
for $\theta>0$ the response coefficient and $\pi^{*}>0$ the \textit{inflation ceiling}. Notice that for $(r,q)\in\mathcal{H}^{G}$, if $r=\mathcal{T}_{2}(r,q)$, then $\pi_{t}\leq\pi^{*}$, $t\geq0$, and this motivates the fact that $\pi^{*}>0$ represents a ceiling on the inflation rates and \textit{not} a target as in $\mathcal{T}_{1}$. Notice that (\ref{eqSecondTaylor}) states that monetary policy only reacts if inflation exceeds the ceiling, which is different from (\ref{eqFirstTaylor}), for which monetary policy reacts to \textit{any} deviation of inflation from its target.

The next example shows how this second Taylor rule bounds inflation rates and does not exclude efficient monetary equilibria if the ceiling is not perfectly adjusted.
\begin{example}\label{ex2}
    I build on Example \ref{ex1}. The economy works precisely as before, although now (\ref{eqInflationTarget}) becomes
    \begin{eqnarray}\label{eqInflationCeiling}
        1+\pi^{*}\geq1+\pi_{0}=\frac{q_{1}}{q_{0}}=p_{1}=f(d,\eta),
    \end{eqnarray}
    for $0<\eta\leq 1+g$. If the ceiling $\pi^{*}>-1$ is set too low, we incur once again in the non-existence problem. Suppose the ceiling is given by $\pi^{*}\geq\pi_{g}$. Notice that the unique efficient monetary equilibria of $\mathcal{E}^{2}$ is 
    \begin{eqnarray*}
        p=(p_{1},p_{2},p_{3},p_{4},\ldots)=\biggr(f(d,1+g),\frac{2+g}{(1+g)d},\frac{2+g}{(1+g)^{2}d},\frac{2+g}{(1+g)^3d},\ldots\biggr).
    \end{eqnarray*}
    Let $\pi=(\pi_{0},\pi_{1},\pi_{2},\ldots)\in\mathbb{R}^{\infty}_{++}$ be any vector of inflation rates that satisfy $1+\pi_{0}=f(d,1+g)$ and $0<\pi_{t}\leq \pi^{*}$. The definition $\pi_{g}$ implies $\pi_{0}=f(d,1+g)-1=\pi_{g}\leq\pi^{*}$, so the vector $\pi$ satisfies the Taylor rule ceiling.

    Then, define $(r,q)\in\mathbb{R}^{\infty}_{+}\times\mathbb{R}^{\infty}_{++}$ as
    \begin{eqnarray*}
        1+r_{t}&=&(1+\pi_{t})\frac{p_{t}}{p_{t+1}}
    \end{eqnarray*}
    for $t\geq1$, and $q_{1}=p_{1}=f(d,1+g)$,
    \begin{eqnarray*}
        q_{t}=p_{t}\prod^{t-1}_{i=1}(1+r_{i}),
    \end{eqnarray*}
    for $t\geq2$. Notice that $\rho(r,q)=p$ and, therefore, $(r,q)\in\mathcal{S}_{\mathcal{T}_{2}}$ leads the economy to the unique efficient monetary equilibrium.  Therefore, with this second Taylor rule, we do not rule out \textit{a priori} the efficient monetary equilibrium when the ceiling is not properly adjusted (i.e., if $\pi^{*}>\pi_{g}$).
\end{example}

Example \ref{ex2} reveals that a not-too-low \textit{inflation ceiling}, rather than an \textit{inflation target}, allows the Government to control inflation (thus preserving the purchasing power of fiat money) and does not preclude the possibility of the economy reaching efficient monetary equilibria in the probable case that the ceiling is not well adjusted (which, in the case of Example \ref{ex2}, boils down to $\pi^{*}>\pi_{g}$).

The second Taylor rule, however, does not guaranty that the economy will reach an efficient equilibrium. A closer look at Example \ref{ex2} reveals that the inefficient equilibria that can emerge are precisely the ones in which the real return rates tend towards zero and, therefore, $\lim_{t\rightarrow\infty}r_{t}-\pi_{t}=0$\footnote{This result relies on the completely skewed endowment distribution of households.}. In order to rule out these inefficient equilibria, let us define the following Taylor rule
\begin{eqnarray}\label{eqThirdTaylor}
    \mathcal{T}_{3}(r,q)=\biggr(r_{1}+\theta_{1}\mathbbm{1}_{(\pi^{*},\infty)}(\pi_{0})+\theta_{2} \mathbbm{1}_{(-\infty,\tau^{*}]}\biggr(\sum_{h\in G_{1}}s^{h}(q_{1},(1+r_{1})^{-1}q_{2})/\sum_{h\in G_{1}}e^{h}_{1}\biggr),\ldots\biggr)
\end{eqnarray}
with $\theta_{1},\theta_{2}>0$ the response coefficients, $\pi^{*}>-1$ the inflation ceiling and $\tau^{*}>0$ the \textit{relative real aggregate savings floor}. This third Taylor rule states that nominal rates will be increased if inflation rises above its ceiling or if the relative real aggregate savings fall below its floor (in this case, the Government increases nominal rates in an attempt to boost real rates and lead the economy to an efficient equilibrium).

Notice that if $r=\mathcal{T}_{3}(r,q)$, then $\pi_{t}\leq\pi^{*}$, $t\geq0$, and 
\begin{eqnarray}\label{eqLowerBoundRealSav}
    \frac{\sum_{h\in G_{t}}s^{h}(q_{t},(1+r_{t})^{-1}q_{t+1})}{\sum_{h\in G_{t}}e^{h}_{t}}>\tau^{*},
\end{eqnarray}
for $t\geq1$. Reasoning as in Example \ref{ex2}, we can once again conclude that, for an inflation target $\pi^{*}>-1$ that is not too low and a real aggregate savings floor $\tau^{*}>0$ that is not too high, $\mathcal{S}_{\mathcal{T}_{3}}\neq\emptyset$. Furthermore, for $p=\rho(r,q)$, $(r,q)\in\mathcal{S}_{\mathcal{T}_{3}}$, (\ref{eqLowerBoundRealSav}) implies
\begin{eqnarray*}
    \inf_{t\geq1}\frac{\sum_{h\in G_{t}}s^{h}(p_{t},p_{t+1})}{\sum_{h\in G_{t}}e^{h}_{t}}\geq\tau^{*},
\end{eqnarray*}
and Lemma \ref{lemmaFirstParetoOptEquilibria} allows us to conclude that $p\in\mathcal{H}^{2PO}$ (i.e., all stable equilibria possibly selected by $\mathcal{T}_{3}(\cdot)$ lead towards efficient monetary equilibria of $\mathcal{E}^{2}$).

The main takeaway from this section can be stated as follows. In prone-to-savings consumption-loan overlapping generations economies, a Taylor rule that is based on a not-too-low inflation ceiling and a not-too-high real aggregate savings floor\footnote{The relevance of monitoring the level of real savings for monetary policy had already been noticed in \textcite[p. 10]{Dognini_2026a}: ``Due to its macroeconomic implications, a result from Section 3 that deserves careful attention is Proposition 7, which shows that the dynamics of real savings per capita completely characterize efficient equilibria in prone-to-savings economies. When it comes to monetary policy, this indicates that it is important to monitor not only the price level (since inflation can make (1) converge and, therefore, lead to an inefficient allocation) but also real savings per capita.''} is capable of: (i) assuring the existence of stable equilibria; (ii) keeping the inflation rate below the ceiling; and, (iii) avoiding inefficient equilibria. Although a inflation-targeting Taylor rule, such as the one given by (\ref{eqFirstTaylor}), satisfies (i) and (ii), it will generally lead the economy towards inefficient equilibria and, therefore, will not satisfy (iii).  
\section{Concluding remarks}\label{sec4}

The interest rate rule in \textcite[p. 202]{Taylor_1993} has a nominal and a real component. The nominal component is an inflation target. The real component is the deviation of GDP from ``trend real GDP'' (which is a difficult number to calculate properly), thus furnishing a sort of GDP target. 

This paper presents an alternative Taylor rule that also has both a nominal and a real component. The nominal component, rather than an inflation target, is an inflation ceiling. The real component, rather than a GDP target, is a relative aggregate real savings floor (which, I believe, can be more easily measured than ``trend real GDP''). In some sense, the ceiling prevents the value-erosion of the monetary base, while the floor sustains real return rates in order to rule out pessimistic self-fulfilling views about the future.

This alternative Taylor rule is, in the consumption-loan models studied in this paper, preferable to the classical inflation targeting rule in the sense that it not only bounds inflation rates but also, under certain conditions (e.g., that the iterative process described after Definition \ref{defTaylorRule} reaches a stable equilibrium), leads the economy towards efficient equilibria (which is not the case in general for the classical rule, as shown by Example \ref{ex1}).

\appendix
\section*{Appendix}\label{appx}

All omitted proofs are stated in this \hyperref[appx]{Appendix} and closely follow those in \textcite[pp. 10-16]{Dognini_2026a}.

\begin{proof}[Proof of Theorem \ref{theoEqSetCompact}]
Let the \textit{set of $j$-sighted equilibria} $\mathcal{V}_{j}\subset\mathbb{R}^{\sum^{j+1}_{i=0}L_{i}}_{++}$, $j\geq1$, of economy $\mathcal{E}$ be given by
\begin{eqnarray}\label{defJSighted}
\mathcal{V}_{j}=\{p\in\mathbb{R}^{\sum^{j+1}_{i=0}L_{i}}_{++}\mid p_{01}=1,Z_{j}(p)=0, (p_{0},p_{1})\in\mathcal{B}_{0}(\sigma_{0}),(p_{j}/p_{j1},p_{j+1}/p_{j1})\in\mathcal{B}_{j}(\sigma_{j})\}.
\end{eqnarray}
First, we must prove that $\mathcal{V}_{j}\subset\mathbb{R}^{\sum^{j+1}_{i=0}L_{i}}_{++}$, $j\geq1$, is compact. Let $(p_{0},\ldots,p_{j+1})\in\mathcal{V}_{j}$, $j\geq1$. Then, $p_{01}=1$, $(p_{0},p_{1})\in\mathcal{B}_{0}(\sigma_{0})$ and $(p_{j}/p_{j1},p_{j+1}/p_{j1})\in\mathcal{B}_{j}(\sigma_{j})$. Market clearing equations at $1\leq t \leq j$ imply the following bound on aggregate demand of generation $G_{t}$ when young,
\begin{eqnarray*}
    \sum_{h\in G_{t}} \Vert x^{h}_{t}(p_{t},p_{t+1})\Vert
    \leq \sum_{h\in G_{t-1}\cup G_{t}} \Vert e^{h}_{t}\Vert\leq\beta \sum_{h\in G_{t}} \Vert e^{h}_{t}\Vert.
\end{eqnarray*}
Also, market clearing at $1\leq t \leq j$ implies the following bound on aggregate demand of generation $G_{t-1}$ when old,
\begin{eqnarray*}
 \sum_{h\in G_{t-1}} \Vert x^{h}_{t}(p_{t-1},p_{t})\Vert
\leq \sum_{h\in G_{t-1}\cup G_{t}}\Vert e^{h}_{t}\Vert\leq \beta  \sum_{h\in G_{t-1}} \Vert e^{h}_{t}\Vert.
\end{eqnarray*}
We conclude that the aggregate demand of generation $G_{t}$, $1\leq t\leq j-1$, satisfies
\begin{eqnarray*}
    \sum_{h\in G_{t}} \Vert x^{h}(p_{t},p_{t+1})\Vert= \sum_{h\in G_{t}} \Vert x^{h}_{t}(p_{t},p_{t+1})\Vert+ \sum_{h\in G_{t}} \Vert x^{h}_{t+1}(p_{t},p_{t+1})\Vert\leq \beta  \sum_{h\in G_{t}}\Vert e^{h}\Vert,
\end{eqnarray*}
and Assumption \ref{assBoundsPrices} implies that $(p_{t}/p_{t1},p_{t+1}/p_{t1})\in\mathcal{B}_{t}(\sigma_{t})$, $1\leq t\leq j-1$. Therefore,  $(p_{t}/p_{t1},p_{t+1}/p_{t1})\in\mathcal{B}_{t}(\sigma_{t})$, for all $0\leq t\leq j$. Since $(p_{0},p_{1})\in\mathcal{B}_{0}(\sigma_{0})$, then
\begin{eqnarray*}
    (\sigma_{0},\ldots,\sigma_{0})\leq p_{0} \leq  (\sigma_{0}^{-1},\ldots,\sigma_{0}^{-1})\\
    (\sigma_{0},\ldots,\sigma_{0})\leq p_{1} \leq  (\sigma_{0}^{-1},\ldots,\sigma_{0}^{-1}),
\end{eqnarray*}
so that $p_{0},p_{1}\in\mathcal{K}_{0}=\mathcal{K}_{1}$. In particular, $\sigma_{0}\leq p_{11}\leq \sigma_{0}^{-1}$. Then, $(p_{1}/p_{11},p_{2}/p_{11})\in\mathcal{B}_{1}(\sigma_{1})$ implies
\begin{eqnarray*}
     (\sigma_{0}\sigma_{1},\ldots,\sigma_{0}\sigma_{1})\leq 
     (\sigma_{1} p_{11},\ldots,\sigma_{1} p_{11})\leq p_{2} \leq  \biggr(\frac{p_{11}}{\sigma_{1}},\ldots,\frac{p_{11}}{\sigma_{1}}\biggr)\leq \biggr(\frac{1}{\sigma_{0}\sigma_{1}},\ldots,\frac{1}{\sigma_{0}\sigma_{1}}\biggr).
\end{eqnarray*}
By induction, we can assert that 
\begin{eqnarray*}
    \biggr(\prod^{t-1}_{j=0}\sigma_{j},\ldots,\prod^{t-1}_{j=0}\sigma_{j}\biggr)\leq p_{t} \leq \biggr(\frac{1}{\prod^{t-1}_{j=0}\sigma_{j}},\ldots,\frac{1}{\prod^{t-1}_{j=0}\sigma_{j}}\biggr),
\end{eqnarray*}
for $1\leq t\leq j+1$. We conclude that $p_{t}\in\mathcal{K}_{t}$, $0\leq t\leq j+1$, and therefore $\mathcal{V}_{j}\subseteq\prod^{j+1}_{t=0}\mathcal{K}_{t}$.

Next, let $\{p^{n}\}_{n\geq1}$ be any sequence with $p^{n}\in\mathcal{V}_{j}$, $n\geq1$, and $\lim_{n\rightarrow\infty}p^{n}=p\in\mathbb{R}^{\sum^{j+1}_{i=0}L_{i}}$. In particular, $p_{01}=\lim_{n\rightarrow\infty} p^{n}_{01}=1$. Also, for $0\leq t\leq j+1$, $p^{n}_{t}\in\mathcal{K}_{t}$, $n\geq1$, which implies that $p_{t}=\lim_{n\rightarrow\infty}p^{n}_{t}\in\mathcal{K}_{t}\subset\mathbb{R}^{L_{t}}_{++}$, since $\mathcal{K}_{t}$ is closed. Then, continuity of Walrasian demand functions on strictly positive prices leads to 
\begin{eqnarray*}
z_{t}(p_{t-1},p_{t},p_{t+1})=\lim_{n\rightarrow\infty} z_{t}(p^{n}_{t-1},p^{n}_{t},p^{n}_{t+1})=0,
\end{eqnarray*}
for $1\leq t\leq j$. Since $\mathcal{B}_{0}(\sigma_{0})$ and $\mathcal{B}_{j}(\sigma_{j})$ are closed and $\lim_{n\rightarrow\infty}p^{n}_{j1}=p_{j1}>0$, we also have $(p_{0},p_{1})\in \mathcal{B}_{0}(\sigma_{0})$ and  $(p_{j}/p_{j1},p_{j+1}/p_{j1})\in \mathcal{B}_{j}(\sigma_{j})$. Therefore, $p\in\mathcal{V}_{j}$ and $\mathcal{V}_{j}$ is a compact set. 

The fact that $\mathcal{V}_{j}$ is non-empty is derived after Theorem 5 from \textcite[119]{ArrowHahn_1971}. By Assumption \ref{assResourceRelated}, every household $h\in\bigcup^{j}_{t=0}G_{t}$ in the ``insulated'' economy formed by all generations until time $j$ is indirectly resource related to every other. Therefore, the aforementioned theorem and Assumption \ref{assBoundsPrices} imply that this economy has an equilibrium $p^{\prime}=(p^{\prime}_{0},\ldots,p^{\prime}_{j+1})\in\mathbb{R}^{\sum^{j+1}_{i=0}L_{i}}_{++}$, with $p^{\prime}_{01}=1$ after a suitable normalization. 

The market clearing equation for this insulated economy in period $t=0$ is given by $\sum_{h\in G_{0}}x^{h}_{0}(p^{\prime}_{0},p^{\prime}_{1})=\sum_{h\in G_{0}}e^{h}_{0}$, and in period $t=j+1$ is given by $\sum_{h\in G_{j}}x^{h}_{j+1}(p^{\prime}_{j},p^{\prime}_{j+1})=\sum_{h\in G_{j}}e^{h}_{j+1}$. Then, $\sum_{h\in G_{0}}\Vert x^{h}_{0}(p^{\prime}_{0},p^{\prime}_{1})\Vert<\beta\sum_{h\in G_{0}} \Vert e^{h}_{0}\Vert$ and $\sum_{h\in G_{j}}\Vert x^{h}_{j+1}(p^{\prime}_{j},p^{\prime}_{j+1})\Vert<\beta\sum_{h\in G_{j}}\Vert e^{h}_{j+1}\Vert$. Also, the market clearing equation in period $1\leq t\leq j$ is given by $z_{t}(p^{\prime}_{t-1},p^{\prime}_{t},p^{\prime}_{t+1})=0$, and, as before, it allows us to limit the demand from generations $G_{0}$ in $t=1$ and $G_{j}$ in $t=j$ by $\beta \sum_{h\in G_{0}} \Vert e^{h}_{1}\Vert$ and $\beta \sum_{h\in G_{j}}\Vert e^{h}_{j}\Vert$, respectively. Then, Assumption \ref{assBoundsPrices} allows us to conclude that $p_{0}\in \mathcal{B}(\sigma_{0})$ and $(p_{j}/p_{j1},p_{j+1}/p_{j1})\in\mathcal{B}(\sigma_{j})$, so that $p^{\prime}\in\mathcal{V}_{j}$ and $\mathcal{V}_{j}$ is non-empty.

Then, let $\mathcal{V}^{\infty}_{j}=\mathcal{V}_{j}
\times\mathbb{R}^{\infty}_{+}$, $j\geq1$. Notice that if $p\in\mathcal{V}^{\infty}_{j+1}$, then $(p_{0},p_{1})\in\mathcal{B}_{0}(\sigma_{0})$, $p_{01}=1$ and $Z_{j+1}(p_{0},\ldots,p_{j+2})=0$. In particular, $Z_{j}(p_{0},\ldots,p_{j+1})=0$ and we can, once again, use Assumption \ref{assBoundsPrices} and the market clearing equations in periods $t=j$ and $t=j+1$ to conclude that $(p_{j}/p_{j1},p_{j+1}/p_{j1})\in \mathcal{B}_{j}(\sigma_{j})$. Then, $(p_{0},\ldots,p_{j+1})\in \mathcal{V}_{j}$ and $p\in\mathcal{V}^{\infty}_{j}$. 

Therefore, $\mathcal{V}^{\infty}_{j+1}\subseteq\mathcal{V}^{\infty}_{j}$, $j\geq1$, and so $\{\mathcal{V}^{\infty}_{j}\cap\mathcal{K}\}_{j\geq1}$ is a nested sequence of sets, with $\mathcal{K}=\prod^{\infty}_{t=0}\mathcal{K}_{t}$ and  
\begin{eqnarray*}
    \mathcal{V}^{\infty}_{j}\cap \mathcal{K}=\mathcal{V}_{j}\times\prod^{\infty}_{t\geq j+2}\mathcal{K}_{t},
\end{eqnarray*}
for $j\geq1$. Then, Tychonoff's Theorem implies $\mathcal{V}^{\infty}_{j}\cap \mathcal{K}$, $j\geq1$, is non-empty and compact. 

First, we prove that $\mathcal{H}\subseteq\lim_{j\rightarrow\infty}(\mathcal{V}^{\infty}_{j}\cap\mathcal{K})$. Let $p\in\mathcal{H}$. Definition \ref{defSetEquilibria} implies $p_{01}=1$ and $(p_{0},p_{1})\in\mathcal{B}_{0}(\sigma_{0})$. By the same reasoning as before, we can use market clearing equations to limit the demand of generation $G_{t}$, $t\geq1$, by $\beta \sum_{h\in G_{t}}\Vert e^{h}\Vert$ and through Assumption \ref{assBoundsPrices} we have
$(p_{t}/p_{t1},p_{t+1}/p_{t1})\in\mathcal{B}_{t}(\sigma_{t})$, $t\geq1$. Also, by the same reasoning, we have $p_{t}\in\mathcal{K}_{t}$, $t\geq0$, thus yielding $p\in\mathcal{K}$. Then, $(p_{0},\ldots,p_{j+1})\in\mathcal{V}_{j}$, $j\geq1$, so that $p\in\mathcal{V}_{j}^{\infty}$, $j\geq1$. We conclude that $p\in \lim_{j\rightarrow\infty}\mathcal{V}^{\infty}_{j}\cap\mathcal{K}$ and, therefore, $\mathcal{H}\subseteq\lim_{j\rightarrow\infty}(\mathcal{V}^{\infty}_{j}\cap\mathcal{K})$. 

Next, if $p\in\lim_{j\rightarrow\infty}(\mathcal{V}^{\infty}_{j}\cap\mathcal{K})$, then $p\in \mathcal{V}^{\infty}_{j}\cap\mathcal{K}$, $j\geq1$, since $\{\mathcal{V}^{\infty}_{j}\cap\mathcal{K}\}_{j\geq1}$ is a nested sequence of sets. Then, (\ref{defJSighted}) implies $p_{01}=1$, $(p_{0},p_{1})\in\mathcal{B}_{0}(\sigma_{0})$ and $z_{t}(p_{t-1},p_{t},p_{t+1})=0$, $t\geq1$. We conclude that $p\in\mathcal{H}$ and, therefore, $\mathcal{H}=\lim_{j\rightarrow\infty}(\mathcal{V}^{\infty}_{j}\cap\mathcal{K})$. 

Since $\mathcal{H}=\lim_{j\rightarrow\infty}(\mathcal{V}^{\infty}_{j}\cap\mathcal{K})$, Cantor's Intersection Theorem implies that $\mathcal{H}$ is non-empty and compact. Furthermore, notice that
\begin{eqnarray*}
    \lim_{j\rightarrow\infty}\mathcal{V}_{j}^{\infty}=\lim_{j\rightarrow\infty}(\mathcal{V}^{\infty}_{j}\cap\mathcal{K})\subseteq\mathcal{K}\implies \mathcal{H}=\lim_{j\rightarrow\infty}\mathcal{V}^{\infty}_{j}\subseteq\mathcal{K}.
\end{eqnarray*}
\end{proof}

\begin{proof}[Proof of Lemma~{\upshape\ref{lemmaFirstParetoOptEquilibria}}]
Since $p\in\mathcal{H}$, $z_{t}(p_{t-1},p_{t},p_{t+1})=0$, $t\geq1$, implies that
\begin{eqnarray}\label{eq1LemmaFirst}
\sum_{h\in G_{t-1}}p_{t}\cdot(x^{h}_{t}(p_{t-1},p_{t})-e^{h}_{t})+\sum_{h\in G_{t}}p_{t}\cdot(x^{h}_{t}(p_{t},p_{t+1})-e^{h}_{t})=0.
\end{eqnarray}
Summing the budget constraints of all households $h\in G_{t}$, $t\geq0$, furnishes, through Walras' law, 
\begin{eqnarray}\label{eq2LemmaFirst}
\sum_{h\in G_{t}}p_{t}\cdot(x^{h}_{t}(p_{t},p_{t+1})-e^{h}_{t})+\sum_{h\in G_{t}}p_{t+1}\cdot(x^{h}_{t+1}(p_{t},p_{t+1})-e^{h}_{t+1})=0. 
\end{eqnarray}
Then, (\ref{eq1LemmaFirst}) and (\ref{eq2LemmaFirst}) imply that 
\begin{eqnarray*}
\sum_{h\in G_{0}}p_{0}\cdot(x^{h}_{0}(p_{0},p_{1})-e^{h}_{0})=\sum_{h\in G_{t}}p_{t}\cdot(x^{h}_{t}(p_{t},p_{t+1})-e^{h}_{t}),
\end{eqnarray*}
for $t\geq1$. Next, Definition \ref{defRealSavings} implies
\begin{eqnarray*}
    \frac{\sum_{h\in G_{t}} s^{h}(p_{t},p_{t+1})}{\sum_{h\in G_{t}}\Vert e^{h}_{t}\Vert}=\frac{\sum_{h\in G_{t}}p_{t}\cdot (e^{h}_{t}-x^{h}_{t}(p_{t},p_{t+1}))}{\Vert p_{t}\Vert\sum_{h\in G_{t}}\Vert e^{h}_{t}\Vert}=\frac{\sum_{h\in G_{0}}p_{0}\cdot(e^{h}_{0}-x^{h}_{0}(p_{0},p_{1}))}{\Vert p_{t}\Vert \sum_{h\in G_{t}}\Vert e^{h}_{t}\Vert}.
\end{eqnarray*}
for $t\geq0$. By Assumption \ref{assCassCriterion}, 
\begin{eqnarray*}
    \mathcal{H}^{PO}=\biggr\{(p_{0},p_{1},\ldots)\in\mathcal{H}\mid \sum^{+\infty}_{t=0}\frac{1}{\Vert p_{t}\Vert\sum_{h\in G_{t}}\Vert e^{h}_{t}\Vert}=+\infty\biggr\}.
\end{eqnarray*}
Since all terms in $\sum^{+\infty}_{t=0}(\Vert p_{t}\Vert\sum_{h\in G_{t}}\Vert e^{h}_{t}\Vert)^{-1}$ are positive, $p\notin\mathcal{H}^{PO}$ implies convergence of the series and, therefore, $\lim_{t\rightarrow\infty}(\Vert p_{t}\Vert\sum_{h\in G_{t}}\Vert e^{h}_{t}\Vert)^{-1}=0$. Finally,
\begin{eqnarray*}
 \lim_{t\rightarrow\infty} \frac{\sum_{h\in G_{t}} s^{h}(p_{t},p_{t+1})}{\sum_{h\in G_{t}}\Vert e^{h}_{t}\Vert}=\sum_{h\in G_{0}}p_{0}\cdot(e^{h}_{0}-x^{h}_{0}(p_{0},p_{1}))\lim_{t\rightarrow\infty}\frac{1}{\Vert p_{t}\Vert\sum_{h\in G_{t}}\Vert e^{h}_{t}\Vert}=0.
\end{eqnarray*}
\end{proof}

\begin{proof}[Proof of Lemma~{\upshape\ref{lemmaSecondParetoOptEquilibria}}]
Let $p\in\mathcal{H}$ and $t\geq0$ be such that $\sum_{h\in G_{t}}s(p_{t},p_{t+1})/\sum_{h\in G_{t}}\Vert e^{h}_{t}\Vert\leq \delta$. If $t=0$, then Definition \ref{defSetEquilibria} implies $p_{01}=1$ and $(p_{0},p_{1})\in\mathcal{B}_{0}(\sigma_{0})$. If $t>0$, market clearing equations and Assumption \ref{assBoundsPrices} imply $(p_{t}/p_{t1},p_{t+1}/p_{t1})\in\mathcal{B}_{t}(\sigma_{t})$. Then, by Assumption \ref{assProneSavingsEconomy}, we have
\begin{eqnarray}\label{eqLemmaSecondParetoOpt}
\frac{\Vert p_{t}\Vert}{\Vert p_{t+1}\Vert} \leq \frac{1}{1+\varepsilon}\frac{\sum_{h\in G_{t+1}}\Vert e^{h}_{t+1}\Vert}{\sum_{h\in G_{t}}\Vert e^{h}_{t}\Vert}.
\end{eqnarray}
Since $p\in\mathcal{H}$, we can write
\begin{eqnarray*}
    \frac{\sum_{h\in G_{t+1}}s^{h}(p_{t+1},p_{t+2})}{\sum_{h\in G_{t+1}}\Vert e^{h}_{t+1}\Vert}&=& \frac{\sum_{h\in G_{t+1}} p_{t+1}\cdot (e^{h}_{t+1}-x^{h}_{t+1}(p_{t+1},p_{t+2}))}{\Vert p_{t+1}\Vert\sum_{h\in G_{t+1}}\Vert e^{h}_{t+1}\Vert}\\
    &=&\frac{\sum_{h\in G_{t}}p_{t+1}\cdot (x^{h}_{t+1}(p_{t},p_{t+1})-e^{h}_{t+1})}{\Vert p_{t+1}\Vert\sum_{h\in G_{t+1}}\Vert e^{h}_{t+1}\Vert}\\
    &=&\frac{\Vert p_{t}\Vert}{\Vert p_{t+1}\Vert}\frac{\sum_{h\in G_{t}}p_{t}\cdot (e^{h}_{t}-x^{h}_{t}(p_{t},p_{t+1}))}{\Vert p_{t}\Vert\sum_{h\in G_{t+1}}\Vert e^{h}_{t+1}\Vert}\\
    &=&\frac{\Vert p_{t}\Vert\sum_{h\in G_{t}}\Vert e^{h}_{t}\Vert}{\Vert p_{t+1}\Vert\sum_{h\in G_{t+1}}\Vert e^{h}_{t+1}\Vert}\frac{ \sum_{h\in G_{t}}s^{h}(p_{t},p_{t+1})}{\sum_{h\in G_{t}}\Vert e^{h}_{t}\Vert}\leq \frac{\delta}{1+\varepsilon}<\delta,
\end{eqnarray*}
where the first equality is derived from Definition \ref{defProneSavings}; the second from $z_{t+1}(p_{t},p_{t+1},p_{t+2})=0$; the third from Walras' law applied to all households from generation $G_{t}$; and the previous to last inequality from (\ref{eqLemmaSecondParetoOpt}). Therefore, $\sum_{h\in G_{t+1}}s(p_{t+1},p_{t+2})/\sum_{h\in G_{t+1}}\Vert e^{h}_{t+1}\Vert\leq \delta$ and, by induction, $\sum_{h\in G_{t+i}}s(p_{t+i},p_{t+i+1})/\sum_{h\in G_{t+i}}\Vert e^{h}_{t+i}\Vert \leq \delta$, for all $i\geq 0$. 

Then, by Assumptions \ref{assBoundsPrices} and \ref{assProneSavingsEconomy}, 
\begin{eqnarray*}
\Vert p_{t+i+1}\Vert \geq \Vert p_{t+i}\Vert(1+\varepsilon)\frac{\sum_{h\in G_{t+i}}\Vert e^{h}_{t+i}\Vert}{\sum_{h\in G_{t+i+1}}\Vert e^{h}_{t+i+1}\Vert}
\end{eqnarray*}
for $i\geq 0$, and
\begin{eqnarray*}
\Vert p_{t+i}\Vert \geq  \Vert p_{t}\Vert(1+\varepsilon)^{i}\prod^{i-1}_{j=0}\frac{\sum_{h\in G_{t+j}}\Vert e^{h}_{t+j}\Vert}{\sum_{h\in G_{t+j+1}}\Vert e^{h}_{t+j+1}\Vert}=\Vert p_{t}\Vert(1+\varepsilon)^{i}\frac{\sum_{h\in G_{t}}\Vert e^{h}_{t}\Vert}{\sum_{h\in G_{t+i}}\Vert e^{h}_{t+i}\Vert}\\
\implies \frac{1}{\Vert p_{t+i} \Vert \sum_{h\in G_{t+i}}\Vert e^{h}_{t+i}\Vert}\leq \frac{1}{\Vert p_{t} \Vert \sum_{h\in G_{t}}\Vert e^{h}_{t}\Vert (1+\varepsilon)^{i}},
\end{eqnarray*}
for $i\geq 1$. Therefore, 
\begin{eqnarray*}
\sum^{+\infty}_{i=0}\frac{1}{\Vert p_{i} \Vert \sum_{h\in G_{i}}\Vert e^{h}_{i}\Vert}\leq\sum^{t-1}_{i=0}\frac{1}{\Vert p_{i}\Vert\sum_{h\in G_{i}}\Vert e^{h}_{i}\Vert}+\frac{1}{\Vert p_{t} \Vert \sum_{h\in G_{t}}\Vert e^{h}_{t}\Vert}\sum^{+\infty}_{i=0}\frac{1}{(1+\varepsilon)^{i}}<+\infty.
\end{eqnarray*}
By Assumption \ref{assCassCriterion}, we conclude that $p\notin\mathcal{H}^{PO}$.
\end{proof}

\begin{proof}[Proof of Proposition~{\upshape\ref{propNecAndSuffConditionParetoOpt}}]
Let $p\in\mathcal{H}$. If $p\notin\mathcal{H}^{PO}$, Lemma \ref{lemmaFirstParetoOptEquilibria} implies 
\begin{eqnarray*}
    \lim_{t\rightarrow\infty}\frac{\sum_{h\in G_{t}}s^{h}(p_{t},p_{t+1})}{\sum_{h\in G_{t}} \Vert e^{h}_{t}\Vert}=0.
\end{eqnarray*}
If $\lim_{t\rightarrow\infty}\sum_{h\in G_{t}}s(p_{t},p_{t+1})/\sum_{h\in G_{t}} \Vert e^{h}_{t}\Vert=0$, then there is $T\geq0$ such that 
\begin{eqnarray*}
\frac{\sum_{h\in G_{T}}s(p_{T},p_{T+1})}{\sum_{h\in G_{T}} \Vert e^{h}_{T}\Vert}\leq\delta,
\end{eqnarray*}
and, by Lemma \ref{lemmaSecondParetoOptEquilibria}, $p\notin\mathcal{H}^{PO}$.
\end{proof}

\begin{proof}[Proof of Theorem~{\upshape\ref{theoExistenceOfParetoOptJME}}]
Let $\mathcal{F}_{k}$, $k\geq0$, be the \textit{finite economy} formed by all generations $G_{t}$, $0\leq t\leq k$, from $\mathcal{E}$ and a single household with an endowment $e^{*}=(0,\ldots,0)\times (\sum_{h\in G_{k+1}}e^{h}_{k+1})\in\mathbb{R}^{L_{0}\times L_{k+1}}_{+}$ and utility function $u^{*}:\mathbb{R}^{L_{0}\times L_{k+1}}_{+}\rightarrow\mathbb{R}$, 
\begin{eqnarray*}
    u^{*}(c_{0},c_{k+1})=\sum^{L_{0}}_{i=1}\log c_{0i}.
\end{eqnarray*}
This ``star'' household can be seen as a time traveler, since he holds an endowment in period $k+1$ but only consumes in period $0$ (i.e., he trades with his parents' generation and with his $(k-1)$-great-grandparents' generation). Furthermore, this time traveler provides a way for ``chopping off a finite segment of the infinite model and then tying the two ends together to form a closed loop'' \parencite[p. 365]{CassYaari_1966}.

Notice that $u^{*}(\cdot)$ is continuous, non-decreasing, semi-strictly quasiconcave and without local maxima. Also, the time traveler is resource related to at least one household from generation $G_{k}$ and, therefore, Assumption \ref{assResourceRelated} implies that all households from $\mathcal{F}_{k}$ are indirectly resource related.

Then, Theorem 5 from \textcite[p. 119]{ArrowHahn_1971} and Assumption \ref{assBoundsPrices} imply the existence of an equilibrium $p^{k}=(p^{k}_{0},\ldots,p^{k}_{k+1})\in\mathbb{R}^{\sum^{k+1}_{i=0}L_{i}}_{++}$ for $\mathcal{F}_{k}$, $k\geq0$, with $p^{k}_{01}=1$ after a suitable normalization. Notice that the market clearing equation in period $t=0$ and Assumption \ref{assBoundsPopAndEndownments} allow us to write
\begin{eqnarray}
\sum_{h\in G_{0}} \Vert x^{h}_{0}(p^{k}_{0},p^{k}_{1})\Vert\leq\sum_{h\in G_{0}}\Vert x^{h}_{0}(p^{k}_{0},p^{k}_{1})\Vert + \Vert x^{*}_{0}(p^{k}_{0},p^{k}_{k+1})\Vert 
= \sum_{h\in G_{0}} \Vert e^{h}_{0}\Vert < \beta  \sum_{h\in G_{0}} \Vert e^{h}_{0}\Vert.\label{eqExist1}
\end{eqnarray}
The market clearing equation in period $t=k+1$ and Assumption \ref{assBoundsPopAndEndownments} also allow us to write
\begin{eqnarray}
\sum_{h\in G_{k}}\Vert  x^{h}_{k+1}(p^{k}_{k},p^{k}_{k+1})\Vert&=&\sum_{h\in G_{k}} \Vert x^{h}_{k+1}(p^{k}_{k},p^{k}_{k+1})\Vert + \Vert x^{*}_{k+1}(p^{k}_{0},p^{k}_{k+1})\Vert \nonumber\\
&=& \sum_{h\in G_{k}} \Vert e^{h}_{k+1}\Vert +\Vert e^{*}_{k+1}\Vert\nonumber\\
&=& \sum_{h\in G_{k}} \Vert e^{h}_{k+1}\Vert +\sum_{h\in G_{k+1}}\Vert e^{h}_{k+1}\Vert\nonumber\\
&\leq& \beta  \sum_{h\in G_{k}} \Vert e^{h}_{k+1}\Vert,\label{eqExist2}
\end{eqnarray}
where the first equality is due to the fact that $x^{*}_{k+1}(p^{k}_{0},p^{k}_{k+1})=0$. Furthermore, market clearing in $t=1$ and $t=k$, (\ref{eqExist1}), (\ref{eqExist2}) and Assumption \ref{assBoundsPrices} imply $(p^{k}_{0},p^{k}_{1})\in \mathcal{B}_{0}(\sigma_{0})$ and $(p^{k}_{k}/p^{k}_{k1},p^{k}_{k+1}/p^{k}_{k1})\in \mathcal{B}_{k}(\sigma_{k})$.

Therefore, reasoning, once again, as in the proof of Theorem \ref{theoEqSetCompact}, we conclude that $p^{k}_{t}\in\mathcal{K}_{t}$, for all $t\leq k+1$. Let $\tilde{p}^{\,k}\in \prod_{t\geq0}\mathcal{K}_{t}$ be an extension of $p^{k}$, $k\geq0$, to $\mathbb{R}^{\infty}$ (i.e., $\tilde{p}^{\,k}_{t}=p^{k}_{t}$, $0\leq t\leq k+1$). Since $\prod_{t\geq0}\mathcal{K}_{t}$ is compact, let $\{\tilde{p}^{\,k_{n}}\}_{n\geq1}$ be a convergent subsequence, with $p=\lim_{n\rightarrow\infty}\tilde{p}^{\,k_{n}}$.

Then, $p^{k}_{01}=1$ and $(p^{k}_{0},p^{k}_{1})\in\mathcal{B}_{0}(\sigma_{0})$, $k\geq0$, imply $p_{01}=1$ and $(p_{0},p_{1})\in\mathcal{B}_{0}(\sigma_{0})$. Also, for every $t\geq1$, there is $N\geq1$ such that $n\geq N$ implies $k_{n}\geq t$. The market clearing equation in period $t\geq1$ for the finite economy $\mathcal{F}_{k_{n}}$, $n\geq N$, implies 
\begin{eqnarray*}
    z_{t}(\tilde{p}^{\,k_{n}}_{t-1},\tilde{p}^{\,k_{n}}_{t},\tilde{p}^{\,k_{n}}_{t+1})=z_{t}(p^{k_{n}}_{t-1},p^{k_{n}}_{t},p^{k_{n}}_{t+1})=0,
\end{eqnarray*}
and, therefore, the continuity of $z_{t}(\cdot)$, $t\geq1$, over strictly positive prices implies
\begin{eqnarray*}
     z_{t}(p_{t-1},p_{t},p_{t+1})=\lim_{n\rightarrow\infty} z_{t}(\tilde{p}^{\,k_{n}}_{t-1},\tilde{p}^{\,k_{n}}_{t},\tilde{p}^{\,k_{n}}_{t+1})=0.
\end{eqnarray*}
We conclude that $p\in\mathcal{H}$. Suppose that $p\notin\mathcal{H}^{PO}$. Lemma \ref{lemmaFirstParetoOptEquilibria} implies 
\begin{eqnarray*}
    \lim_{t\rightarrow\infty}\frac{\sum_{h\in G_{t}}s^{h}(p_{t},p_{t+1})}{\sum_{h\in G_{t}}\Vert e^{h}_{t}\Vert}=0,
\end{eqnarray*}
and, therefore, there is $T\geq1$ such that
\begin{eqnarray*}
    \frac{\sum_{h\in G_{T}}s^{h}(p_{T},p_{T+1})}{\sum_{h\in G_{T}}\Vert e^{h}_{T}\Vert}<\frac{\delta}{2}.
\end{eqnarray*}
Convergence on the product topology implies that there is $m\geq 1$ such that $k_{m}\geq T$ and 
\begin{eqnarray}\label{eqExistSub}
    \frac{\sum_{h\in G_{T}}s^{h}(\tilde{p}^{\,k_{m}}_{T},\tilde{p}^{\,k_{m}}_{T+1})}{\sum_{h\in G_{T}}\Vert e^{h}_{T}\Vert}=\frac{\sum_{h\in G_{T}}s^{h}(p^{k_{m}}_{T},p^{k_{m}}_{T+1})}{\sum_{h\in G_{T}}\Vert e^{h}_{T}\Vert}\leq\delta.
\end{eqnarray}
If $k_{m}=T$, then (\ref{eqExistSub}) implies 
\begin{eqnarray}\label{eqExist3}
    \frac{\sum_{h\in G_{k_{m}}} s^{h}(p^{k_{m}}_{k_{m}},p^{k_{m}}_{k_{m}+1})}{\sum_{h\in G_{k_{m}}}\Vert e^{h}_{k_{m}}\Vert}\leq\delta.
\end{eqnarray}
If $k_{m}>T$, we can write
\begin{eqnarray*}
    \frac{\sum_{h\in G_{T+1}}s^{h}(p^{k_{m}}_{T+1},p^{k_{m}}_{T+2})}{\sum_{h\in G_{T+1}}\Vert e^{h}_{T+1}\Vert}&=& \frac{\sum_{h\in G_{T+1}} p^{k_{m}}_{T+1}\cdot (e^{h}_{T+1}-x^{h}_{T+1}(p^{k_{m}}_{T+1},p^{k_{m}}_{T+2}))}{\Vert p^{k_{m}}_{T+1}\Vert\sum_{h\in G_{T+1}}\Vert e^{h}_{T+1}\Vert}\\
    &=&\frac{\sum_{h\in G_{T}}p^{k_{m}}_{T+1}\cdot (x^{h}_{T+1}(p^{k_{m}}_{T},p^{k_{m}}_{T+1})-e^{h}_{T+1})}{\Vert p^{k_{m}}_{T+1}\Vert\sum_{h\in G_{T+1}}\Vert e^{h}_{T+1}\Vert}\\
    &=&\frac{\Vert p^{k_{m}}_{T}\Vert}{\Vert p^{k_{m}}_{T+1}\Vert}\frac{\sum_{h\in G_{T}}p_{T}\cdot (e^{h}_{T}-x^{h}_{T}(p^{k_{m}}_{T},p^{k_{m}}_{T+1}))}{\Vert p^{k_{m}}_{T}\Vert\sum_{h\in G_{T+1}}\Vert e^{h}_{T+1}\Vert}\\
    &=&\frac{\Vert p^{k_{m}}_{T}\Vert\sum_{h\in G_{T}}\Vert e^{h}_{T}\Vert}{\Vert p^{k_{m}}_{T+1}\Vert\sum_{h\in G_{T+1}}\Vert e^{h}_{T+1}\Vert}\frac{\sum_{h\in G_{T}}s^{h}(p^{k_{m}}_{T},p^{k_{m}}_{T+1})}{\sum_{h\in G_{T}}\Vert e^{h}_{T}\Vert}\leq \frac{\delta}{1+\varepsilon}<\delta,
\end{eqnarray*}
where the first equality is derived from Definition \ref{defRealSavings}; the second from $z_{T+1}(p_{T},p_{T+1},p_{T+2})=0$; the third from Walras' law applied to all households from generation $G_{T}$; and the penultimate inequality from Definition \ref{defProneSavings} and (\ref{eqExistSub}). Then, 
\begin{eqnarray*}
     \frac{\sum_{h\in G_{T+1}}s^{h}(p^{k_{m}}_{T+1},p^{k_{m}}_{T+2})}{\sum_{h\in G_{T+1}}\Vert e^{h}_{T+1}\Vert}\leq \delta,
\end{eqnarray*}
and, by induction, we conclude that (\ref{eqExist3}) remains valid. 

Next, considering the finite economy $\mathcal{F}_{k_{m}}$, we can write
\begin{eqnarray*}
    \frac{\sum_{h\in G_{k_{m}}} s^{h}(p^{k_{m}}_{k_{m}},p^{k_{m}}_{k_{m}+1})}{\sum_{h\in G_{k_{m}}}\Vert e^{h}_{k_{m}}\Vert}&=&\frac{\sum_{h\in G_{k_{m}}} p^{k_{m}}_{k_{m}}\cdot (e^{h}_{k_{m}}-x^{h}_{k_{m}}(p^{k_{m}}_{k_{m}},p^{k_{m}}_{k_{m}+1}))}{\Vert p^{k_{m}}_{k_{m}}\Vert \sum_{h\in G_{k_{m}}}\Vert e^{h}_{k_{m}}\Vert}\\
    &=&\frac{\sum_{h\in G_{k_{m}}} p^{k_{m}}_{k_{m}+1}\cdot (x^{h}_{k_{m}+1}(p^{k_{m}}_{k_{m}},p^{k_{m}}_{k_{m}+1})-e^{h}_{k_{m}+1})}{\Vert p^{k_{m}}_{k_{m}}\Vert \sum_{h\in G_{k_{m}}}\Vert e^{h}_{k_{m}}\Vert}\\
    &=&\frac{\sum_{h\in G_{k_{m}+1}} p^{k_{m}}_{k_{m}+1}\cdot e^{h}_{k_{m}+1}}{\Vert p^{k_{m}}_{k_{m}}\Vert \sum_{h\in G_{k_{m}}}\Vert e^{h}_{k_{m}}\Vert},
\end{eqnarray*}
where the first equality is derived from Definition \ref{defRealSavings}; the second from Walras' law applied to all households from generation $G_{k_{m}}$; and the third from market clearing at $t=k_{m}+1$ and the fact that $x^{*}_{k_{m}+1}(p^{k_{m}}_{k_{m}},p^{k_{m}}_{k_{m}+1})=0$. Then, (\ref{eqExist3}) implies
\begin{eqnarray*}
    \frac{\sum_{h\in G_{k_{m}+1}} p^{k_{m}}_{k_{m}+1}\cdot e^{h}_{k_{m}+1}}{\Vert p^{k_{m}}_{k_{m}}\Vert \sum_{h\in G_{k_{m}}}\Vert e^{h}_{k_{m}}\Vert}\leq \delta.
\end{eqnarray*}
 However, Definition \ref{defProneSavings} and (\ref{eqExist3}) imply
\begin{eqnarray*}
    \frac{\Vert p^{k_{m}}_{k_{m}}\Vert \sum_{h\in G_{k_{m}}}\Vert e^{h}_{k_{m}}\Vert}{\Vert p^{k_{m}}_{k_{m}+1}\Vert \sum_{h\in G_{k_{m}+1}}\Vert e^{h}_{k_{m}+1}\Vert}\leq \frac{1}{1+\varepsilon},
\end{eqnarray*}
and we can write
\begin{eqnarray}\label{eqExist4}
    \frac{p^{k_{m}}_{k_{m}+1}}{\Vert p^{k_{m}}_{k_{m}+1}\Vert}\cdot \frac{\sum_{h\in G_{k_{m}+1}} e^{h}_{k_{m}+1}}{\sum_{h\in G_{k_{m}+1}}\Vert e^{h}_{k_{m}+1}\Vert}<\delta.
\end{eqnarray}
Due to Assumption \ref{assUniformBoundsDeltaAndL}, let $\sigma=\inf_{t\geq0}\sigma_{t}>0$ and $L=\sup_{t\geq0}L_{t}<\infty$. Since $(p^{k_{m}}_{k_{m}}/p^{k_{m}}_{k_{m},1},p^{k_{m}}_{k_{m}+1}/p^{k_{m}}_{k_{m},1})\in B_{k_{m}}(\sigma_{k_{m}})\subseteq B_{k_{m}}(\sigma)$, we have
\begin{eqnarray*}
     \frac{p^{k_{m}}_{k_{m}+1}}{\Vert p^{k_{m}}_{k_{m}+1}\Vert}\cdot \frac{\sum_{h\in G_{k_{m}+1}} e^{h}_{k_{m}+1}}{\sum_{h\in G_{k_{m}+1}}\Vert e^{h}_{k_{m}+1}\Vert}&=&\frac{p^{k_{m}}_{k_{m}+1}/p^{k_{m}}_{k_{m},1}}{\Vert p^{k_{m}}_{k_{m}+1}/p^{k_{m}}_{k_{m},1}\Vert}\cdot \frac{\sum_{h\in G_{k_{m}+1}} e^{h}_{k_{m}+1}}{\sum_{h\in G_{k_{m}+1}}\Vert e^{h}_{k_{m}+1}\Vert}\\
     &\geq&\frac{(\sigma,\ldots,\sigma)}{L_{k_{m}+1}\sigma^{-1}}\cdot \frac{\sum_{h\in G_{k_{m}+1}} e^{h}_{k_{m}+1}}{\sum_{h\in G_{k_{m}+1}}\Vert e^{h}_{k_{m}+1}\Vert}\\
     &=&\frac{\sigma^{2}}{L_{k_{m}+1}}>\frac{\sigma^{2}}{L},
\end{eqnarray*}
so that (\ref{eqExist4}) implies $\sigma^{2}/L<\delta$, absurd, since we can assume without loss of generality $\delta<\sigma^{2}/L$. We conclude that $p\in\mathcal{H}^{PO}$ and, therefore, $\mathcal{H}^{PO}$ is non-empty.

Next, let $\{p^{n}\}_{n\geq1}$ be a sequence in $\mathcal{H}^{PO}$ with $\lim_{n\rightarrow\infty}p^{n}=p\in\mathbb{R}^{\infty}_{++}$. Since $\mathcal{H}^{PO}\subseteq\mathcal{H}$ and $\mathcal{H}$ is compact by Theorem \ref{theoEqSetCompact}, $p\in\mathcal{H}$. If $p\notin\mathcal{H}^{PO}$, Proposition \ref{propNecAndSuffConditionParetoOpt} implies that $\lim_{t\rightarrow\infty}\sum_{h\in G_{t}}s^{h}(p_{t},p_{t+1})/\sum_{h\in G_{t}}\Vert e^{h}_{t}\Vert=0$ and, therefore,  there is $T\geq0$ such that 
\begin{eqnarray*}
    \frac{\sum_{h\in G_{T}}s^{h}(p_{T},p_{T+1})}{\sum_{h\in G_{t}}\Vert e^{h}_{T}\Vert}<\delta/2.
\end{eqnarray*}
Convergence on the product topology implies that there is $N\geq1$ such that 
\begin{eqnarray*}
    \frac{\sum_{h\in G_{T}}s^{h}(p^{N}_{T},p^{N}_{T+1})}{\sum_{h\in G_{T}} e^{h}_{T}}\leq\delta.
\end{eqnarray*}
Then, Lemma \ref{lemmaSecondParetoOptEquilibria} applied to $p^{N}\in\mathcal{H}$ in period $T\geq0$ implies that $p^{N}\notin\mathcal{H}^{PO}$, absurd. Therefore, $p\in\mathcal{H}^{PO}$ and $\mathcal{H}^{PO}$ is closed. We conclude that $\mathcal{H}^{PO}$, as a closed subset of the compact set $\mathcal{H}$, is itself compact. Lastly, Lemma \ref{lemmaSecondParetoOptEquilibria} implies that
\begin{eqnarray*}
    \frac{\sum_{h\in G_{0}}s^{h}(p_{0},p_{1})}{\sum_{h\in G_{0}} \Vert e^{h}_{0}\Vert}>\delta,
\end{eqnarray*}
for $p\in\mathcal{H}^{PO}$.
\end{proof}

\begin{proof}[Proof of Theorem~{\upshape\ref{theoClassModelExistence}}]
Our goal is to apply Theorem \ref{theoExistenceOfParetoOptJME} and, in order to do so, it is necessary to ``extend'' generation $G^{2}_{0}$ of the economy $\mathcal{E}^{2}$ back to period $t=0$. Let $\zeta\geq1$ and $\mathcal{E}(\zeta)$ be the extended economy we envision, with a single perishable commodity in period $t=0$ (i.e., $L_{0}=1$). Let all generations $G_{t}(\zeta)$, $t\geq1$, from $\mathcal{E}(\zeta)$ coincide with those from $\mathcal{E}^{2}_{\geq1}$. For all households $h\in G^{2}_{0}$, we define $h^{\prime}\in G_{0}(\zeta)$ by a utility function $u^{h^{\prime}}:\mathbb{R}^{1+L_{1}}_{+}\rightarrow\mathbb{R}$, given by
\begin{eqnarray*}
    u^{h^{\prime}}(c_{0},c_{1})=u^{h}(c_{1}),
\end{eqnarray*}
and a nonzero endowment $e^{h^{\prime}}=(e^{h^{\prime}}_{0},e^{h^{\prime}}_{1})=(m^{h},e^{h}_{1})\in\mathbb{R}^{1+L_{1}}_{+}$. Clearly, for all $h^{\prime}\in G_{0}(\zeta)$, $u^{h^{\prime}}(\cdot)$ is continuous, non-decreasing, semi-strictly quasiconcave and without local maxima, and $\sum_{h^{\prime}\in G_{0}(\zeta)}e^{h^{\prime}}\in\mathbb{R}^{1+L_{1}}_{++}$. Utility maximization implies $x^{h^{\prime}}_{0}(p_{0},p_{1})=0$ and, therefore,
\begin{eqnarray*}
    s^{h^{\prime}}(p_{0},p_{1})=e^{h^{\prime}}_{0}=m^{h},
\end{eqnarray*}
for $(p_{0},p_{1})\in\mathbb{R}^{1+L_{1}}_{++}$ and $h^{\prime}\in G_{0}(\zeta)$. In particular, we have
\begin{eqnarray}\label{eqExisClassic1}
    \sum_{h^{\prime}\in G_{0}(\zeta)}s^{h^{\prime}}(p_{0},p_{1})=\sum_{h^{\prime}\in G_{0}(\zeta)}e^{h^{\prime}}_{0}=\sum_{h\in G_{0}^{2}}m^{h}=1.
\end{eqnarray}

Since $\mathcal{E}^{2}_{\geq1}$ satisfies Assumption \ref{assProneSavingsEconomy}, let $\varepsilon,\delta>0$ be given by Definition \ref{defProneSavings}. Assumption \ref{assBoundsPopAndEndownments} implies 
\begin{eqnarray}\label{eqNonBindingLimit}
    \alpha_{\min}<\frac{\sum_{h\in G_{1}(\zeta)}\Vert e^{h}_{1} \Vert}{\sum_{h^{\prime}\in G_{0}(\zeta)}\Vert e^{h^{\prime}}_{1}\Vert}=\frac{\sum_{h\in G^{2}_{1}}\Vert e^{h}_{1} \Vert}{\sum_{h\in G^{2}_{0}}\Vert e^{h}_{1}\Vert}<\alpha_{\max},
\end{eqnarray}
and that we can assume, without loss of generality, $\delta<1/2$.

Notice that $\sum_{h^{\prime}\in G_{0}(\zeta)} x^{h^{\prime}}_{0}(p_{0},p_{1})=0$ implies that, with these households only, we will not be able to fulfill Assumption \ref{assBoundsPrices}, since no matter how low $p_{0}=p_{m}\in\mathbb{R}_{++}$ becomes relative to $p_{1}\in\mathbb{R}^{L_{1}}_{++}$, the demand for money will never skyrocket.

Then, we must add a single household $h^{*}_{\zeta}$ to $G_{0}(\zeta)$ with a utility function $u^{h^{*}_{\zeta}}:\mathbb{R}^{1+L_{1}}_{+}\rightarrow\mathbb{R}$, given by
\begin{eqnarray*}
    u^{h^{*}_{\zeta}}(c_{0},c_{1})=f(\zeta)\log c_{0}+\sum^{L_{1}}_{i=1}\log c_{1i},
\end{eqnarray*}
with 
\begin{eqnarray*}
    f(\zeta)=\frac{\zeta L_{1}\sum_{h\in G_{1}}\Vert e^{h}_{1}\Vert}{2\theta(1+\varepsilon)}>0,
\end{eqnarray*}
and an endowment $e^{h^{*}_{\zeta}}=(1,\theta/\zeta,\ldots,\theta/\zeta)\in\mathbb{R}^{1+L_{1}}_{++}$, where $\theta>0$ can be chosen, due to (\ref{eqNonBindingLimit}) and $\zeta\geq1$, so that
\begin{eqnarray}\label{eqTheta}
    \frac{\sum_{h\in G_{1}(\zeta)}\Vert e^{h}_{1} \Vert}{\Vert e^{h^{*}_{\zeta}}_{1}\Vert+\sum_{h^{\prime}\in G_{0}(\zeta)}\Vert e^{h^{\prime}}_{1}\Vert}=\frac{\sum_{h\in G^{2}_{1}}\Vert e^{h}_{1} \Vert}{L_{1}\theta/\zeta+\sum_{h\in G^{2}_{0}}\Vert e^{h^{\prime}}_{1}\Vert}>\alpha_{\min},
\end{eqnarray}
and
\begin{eqnarray}\label{eqTheta2}
    \frac{\Vert e^{h^{*}_{\zeta}}_{1}\Vert+\sum_{h\in G_{1}(\zeta)}\Vert e^{h}_{1} \Vert}{\sum_{h^{\prime}\in G_{0}(\zeta)}\Vert e^{h^{\prime}}_{1}\Vert}=\frac{L_{1}\theta/\zeta+\sum_{h\in G^{2}_{1}}\Vert e^{h}_{1} \Vert}{\sum_{h\in G^{2}_{0}}\Vert e^{h^{\prime}}_{1}\Vert}<\alpha_{\max}.
\end{eqnarray}
This star-household provides a strictly positive demand for money
\begin{eqnarray*}
    x^{h^{*}_{\zeta}}(p_{0},p_{1})=\frac{p_{0}+\Vert p_{1}\Vert\theta/\zeta}{f(\zeta)+L_{1}}\biggr(\frac{f(\zeta)}{p_{0}},\frac{1}{p_{11}},\ldots,\frac{1}{p_{1L_{1}}}\biggr),
\end{eqnarray*}
so that
\begin{eqnarray}\label{eqExisClassic2}
    s^{h^{*}_{\zeta}}(p_{0},p_{1})=1-x^{h^{*}_{\zeta}}_{0}(p_{0},p_{1})=\frac{p_{0}L_{1}-f(\zeta)\Vert p_{1}\Vert\theta/\zeta}{p_{0}(f(\zeta)+L_{1})},
\end{eqnarray}
for $(p_{0},p_{1})\in\mathbb{R}^{1+L_{1}}_{++}$. Notice that
\begin{eqnarray}\label{eqExisClassic3}
    \frac{p_{0}L_{1}-f(\zeta)\Vert p_{1}\Vert\theta/\zeta}{p_{0}(f(\zeta)+L_{1})}+1\leq \delta \implies
    \frac{p_{0}}{\Vert p_{1}\Vert}\leq \frac{f(\zeta)\theta/\zeta}{L_{1}+(1-\delta)(f(\zeta)+L_{1})}.
\end{eqnarray}
Since
\begin{eqnarray*}
    \frac{f(\zeta)\theta/\zeta}{L_{1}+(1-\delta)(f(\zeta)+L_{1})}\leq \frac{f(\zeta)\theta}{\zeta L_{1}}=\frac{\sum_{h\in G_{1}}\Vert e^{h}_{1}\Vert}{2(1+\varepsilon)},
\end{eqnarray*}
we conclude, through (\ref{eqExisClassic1}), (\ref{eqExisClassic2}) and (\ref{eqExisClassic3}), that
\begin{eqnarray}\label{eqExisClassic4}
    s^{h^{*}_{\zeta}}(p_{0},p_{1})+\sum_{h^{\prime}\in G_{0}(\zeta)}s^{h^{\prime}}(p_{0},p_{1})\leq \delta\implies\frac{p_{0}}{\Vert p_{1}\Vert}\leq \frac{1}{1+\varepsilon}\frac{\sum_{h\in G_{1}(\zeta)}\Vert e^{h}_{1}\Vert}{e^{h^{*}_{\zeta}}_{0}+\sum_{h^{\prime}\in G_{0}(\zeta)} e^{h^{\prime}}_{0}},
\end{eqnarray}
for $(p_{0},p_{1})\in\mathbb{R}^{1+L_{1}}_{++}$.

We proceed to show that $\mathcal{E(\zeta)}$ satisfies Assumptions \ref{assBoundsPopAndEndownments}--\ref{assBoundsPrices} and \ref{assProneSavingsEconomy}-\ref{assUniformBoundsDeltaAndL}.

First, since all generations $G_{t}(\zeta)$, $t\geq1$, from the extended economy $\mathcal{E}(\zeta)$ coincide with those from $\mathcal{E}^{2}_{\geq1}$, and $\mathcal{E}^{2}_{\geq1}$ satisfies Assumption \ref{assBoundsPopAndEndownments}, we only need to check generation $G_{0}(\zeta)$. Then, (\ref{eqTheta}) implies the following bounds on the intraperiod growth rate between endowments at $t=1$
\begin{eqnarray*}
    \alpha_{\min}<\frac{\sum_{h\in G_{1}(\zeta)}\Vert e^{h}_{1} \Vert}{\Vert e^{h^{*}_{\zeta}}_{1}\Vert+\sum_{h^{\prime}\in G_{0}(\zeta)}\Vert e^{h^{\prime}}_{1}\Vert}<\frac{\sum_{h\in G^{2}_{1}}\Vert e^{h}_{1} \Vert}{\sum_{h\in G^{2}_{0}(\zeta)}\Vert e^{h}_{1}\Vert}<\alpha_{\max},
\end{eqnarray*}
so that Assumption \ref{assBoundsPopAndEndownments} is satisfied. Assumption \ref{assResourceRelated} is also satisfied by $\mathcal{E}(\zeta)$, since $h^{*}_{\zeta}$ is resource related to at least one $h^{\prime}\in G_{0}(\zeta)$ and households from $\mathcal{E}^{2}$ are indirectly resource related.

Let $\sigma_{t}\in(0,1)$, $t\geq1$, be given by Assumption \ref{assBoundsPrices} when considering the economy $\mathcal{E}^{2}_{\geq1}$, so that 
\begin{eqnarray}\label{eqExisClassic5}
\biggr(\frac{p_{t}}{p_{t1}},\frac{p_{t+1}}{p_{t1}}\biggr)\notin\mathcal{B}_{t}(\sigma_{t})\implies  \sum_{h\in G_{t}(\zeta)}\Vert x^{h}(p_{t},p_{t+1})\Vert=\sum_{h\in G^{2}_{t}} \Vert x^{h}(p_{t},p_{t+1})\Vert>\beta \sum_{h\in G^{2}_{t}}\Vert e^{h}\Vert,
\end{eqnarray}
for $(p_{t},p_{t+1})\in\mathbb{R}^{L_{t}+L_{t+1}}_{++}$. Next, notice that there is $\sigma_{0}(\zeta)\in(0,1)$ such that
\begin{eqnarray*}
    \biggr(1,\frac{p_{1}}{p_{0}}\biggr)\notin\mathcal{B}_{0}(\sigma_{0}(\zeta))\implies \Vert x^{h^{*}_{\zeta}}(p_{0},p_{1})\Vert>\beta (\Vert e^{h^{*}_{\zeta}}\Vert+\sum_{h^{\prime}\in G_{0}(\zeta)}\Vert e^{h^{\prime}}\Vert),
\end{eqnarray*}
for $(p_{0},p_{1})\in\mathbb{R}^{1+L_{1}}_{++}$. Since
\begin{eqnarray*}
    \Vert x^{h^{*}_{\zeta}}(p_{0},p_{1})\Vert+ \sum_{h^{\prime}\in G_{0}(\zeta)}\Vert x^{h^{\prime}}(p_{0},p_{1})\Vert>\Vert x^{h^{*}_{\zeta}}(p_{0},p_{1})\Vert,
\end{eqnarray*}
we have
\begin{eqnarray}\label{eqExisClassic6}
    \biggr(1,\frac{p_{1}}{p_{0}}\biggr)\notin\mathcal{B}_{0}(\sigma_{0}(\zeta))\implies \Vert x^{h^{*}_{\zeta}}(p_{0},p_{1})\Vert+ \sum_{h^{\prime}\in G_{0}(\zeta)}\Vert x^{h^{\prime}}(p_{0},p_{1})\Vert>\beta (\Vert e^{h^{*}_{\zeta}}\Vert+\sum_{h^{\prime}\in G_{0}(\zeta)}\Vert e^{h^{\prime}}\Vert),
\end{eqnarray}
for $(p_{0},p_{1})\in\mathbb{R}^{1+L_{1}}_{++}$. We conclude, by (\ref{eqExisClassic5}) and (\ref{eqExisClassic6}), that $\mathcal{E}(\zeta)$ satisfies Assumption \ref{assBoundsPrices}.

Since all generations $G_{t}(\zeta)$, $t\geq1$, from $\mathcal{E}(\zeta)$ coincide with those from $\mathcal{E}^{2}_{\geq1}$, and $\mathcal{E}^{2}_{\geq1}$ satisfies Assumption \ref{assProneSavingsEconomy}, we only need, once again, to check generation $G_{0}(\zeta)$. For $(1,p_{1}/p_{0})\in \mathcal{B}_{0}(\sigma_{0}(\zeta))$, $(p_{0},p_{1})\in\mathbb{R}^{1+L_{1}}_{++}$, the result follows directly from (\ref{eqExisClassic4}). Furthermore, Assumption \ref{assUniformBoundsDeltaAndL} is clearly satisfied by $\mathcal{E}(\zeta)$.

Lastly, notice that we do not need to verify whether Assumption \ref{assCassCriterion} is satisfied by $\mathcal{E}(\zeta)$. If it is not, then the subset $\mathcal{H}^{PO}(\zeta)$ must be read not as the ``subset of Pareto optimal equilibria,'' but as the ``subset of equilibria that satisfy the Cass criterion.'' 

This possible different interpretation, however, does not prevent us from applying Theorem \ref{theoExistenceOfParetoOptJME} to $\mathcal{E}(n)$, $n\geq1$, to find a sequence $\{p^{n}\}_{n\geq1}$ such that $p^{n}\in\mathcal{H}^{PO}(n)$. Then, the market clearing equation in period $t=1$ allows us to write
\begin{eqnarray*}
     \sum_{h\in G^{2}_{0}}\Vert x^{h}_{1}(1,p^{n}_{1})\Vert &=&  \sum_{h^{\prime}\in G_{0}(n)}\Vert x^{h^{\prime}}_{1}(p^{n}_{0},p^{n}_{1})\Vert\\
    &\leq& \Vert e^{h^{*}_{n}}_{1}\Vert+\sum_{h^{\prime}\in G_{0}(n)}\Vert e^{h^{\prime}}_{1}\Vert+\sum_{h\in G_{1}(n)}\Vert e^{h}_{1}\Vert\\
    &=&\Vert e^{h^{*}_{n}}_{1}\Vert+\sum_{h\in G^{2}_{0}} \Vert e^{h}_{1}\Vert+\sum_{h\in G^{2}_{1}}\Vert e^{h}_{1}\Vert\\
    &\leq&\beta\sum_{h\in G^{2}_{0}} \Vert e^{h}_{1}\Vert,
\end{eqnarray*}
with this last inequality due to (\ref{eqTheta2}). By assumption of the theorem, this implies $p^{n}_{11}>\lambda>0$, $n\geq1$, and therefore 
\begin{eqnarray}\label{eqBoundPn11}
   0<\frac{1}{p^{n}_{11}}<\frac{1}{\lambda},
\end{eqnarray}
for $n\geq1$. Next, let $q^{n}=(q^{n}_{1},q^{n}_{2},\ldots)=(p^{n}_{1}/p^{n}_{11},p^{n}_{2}/p^{n}_{11},\ldots)\in\mathbb{R}^{\infty}_{++}$, $n\geq1$, so that $q^{n}_{11}=1$, $p^{n}=(1,p^{n}_{11}q^{n})$, and
\begin{eqnarray}\label{eqExisClassic7}
    \sum^{\infty}_{t=1}\frac{1}{\Vert q^{n}_{t}\Vert\sum_{h\in G^{2}_{t}}\Vert e^{h}_{t}\Vert}=p^{n}_{11}\sum^{\infty}_{t=1}\frac{1}{\Vert p^{n}_{t}\Vert \sum_{h\in G^{2}_{t}}\Vert e^{h}_{t}\Vert}=+\infty,
\end{eqnarray}
for $n\geq1$. Notice that the market clearing equation in period $t=1$ for the economy $\mathcal{E}(n)$ and the homogeneity of demand imply
\begin{eqnarray}
    \sum_{h\in G^{2}_{1}} \Vert x^{h}_{1}(q^{n}_{1},q^{n}_{2})\Vert&=&\sum_{h\in G_{1}(n)} \Vert x^{h}_{1}(q^{n}_{1},q^{n}_{2})\Vert\nonumber\\
    &\leq& \Vert x^{h^{*}_{n}}_{1}(p^{n}_{0},p^{n}_{1})\Vert+\sum_{h^{\prime}\in G_{0}(n)} \Vert x^{h^{\prime}}_{1}(p^{n}_{0},p^{n}_{1})\Vert+\sum_{h\in G_{1}(n)} \Vert x^{h}_{1}(p^{n}_{1},p^{n}_{2})\Vert\nonumber\\
    &=&\Vert e^{h^{*}_{n}}_{1}\Vert+\sum_{h\in G^{2}_{0}} \Vert e^{h}_{1}\Vert+\sum_{h\in G^{2}_{1}} \Vert e^{h}_{1}\Vert\nonumber\\
    &\leq&\beta \sum_{h\in G^{2}_{1}} \Vert e^{h}_{1}\Vert. \label{eqFirstBoundG1},
\end{eqnarray}
with this last inequality due to (\ref{eqTheta}). Furthermore, the market clearing equation in period $t\geq2$ and the homogeneity of demand imply
\begin{eqnarray}
    \sum_{h\in G^{2}_{t-1}}x^{h}_{t}(q^{n}_{t-1},q^{n}_{t})+\sum_{h\in G^{2}_{t}}x^{h}_{t}(q^{n}_{t},q^{n}_{t+1})&=&\sum_{h\in G_{t-1}(n)}x^{h}_{t}(p^{n}_{t-1},p^{n}_{t})+\sum_{h\in G_{t}(n)}x^{h}_{t}(p^{n}_{t},p^{n}_{t+1})\nonumber\\
    &=&\sum_{h\in G_{t-1}(n)}e^{h}_{t}+\sum_{h\in G_{t}(n)}e^{h}_{t}\nonumber\\
    &=&\sum_{h\in G^{2}_{t-1}}e^{h}_{t}+\sum_{h\in G^{2}_{t}}e^{h}_{t}.\label{eqExisClassic8}
\end{eqnarray}
In particular, for $t=2$ we have
\begin{eqnarray}\label{eqSecondBoundG1}
     \sum_{h\in G^{2}_{1}}\Vert x^{h}_{2}(q^{n}_{1},q^{n}_{2})\Vert\leq \beta \sum_{h\in G^{2}_{1}}\Vert e^{h}_{2}\Vert.
\end{eqnarray}
Therefore, (\ref{eqFirstBoundG1}) and (\ref{eqSecondBoundG1}) imply
\begin{eqnarray*}
    \sum_{h\in G^{2}_{1}}\Vert x^{h}(q^{n}_{1},q^{n}_{2})\Vert\leq \beta \sum_{h\in G^{2}_{1}}\Vert e^{h}\Vert,
\end{eqnarray*}
and, by Assumption \ref{assBoundsPrices}, we have 
\begin{eqnarray}\label{eqExisClassic9}
    (q^{n}_{1},q^{n}_{2})\in\mathcal{B}_{1}(\sigma_{1}),
\end{eqnarray}
for $n\geq1$.

Let $\mathcal{H}^{PO}_{\geq1}$ be the subset of Pareto optimal equilibria from $\mathcal{E}^{2}_{\geq1}$. Notice that (\ref{eqExisClassic7}), (\ref{eqExisClassic8}), and (\ref{eqExisClassic9}) imply $q^{n}\in\mathcal{H}^{PO}_{\geq1}$ (there is a one-period relabeling of time in this assertion since $\mathcal{E}^{2}_{\geq1}$ ``starts'' with generation $G^{2}_{1}$). Theorem \ref{theoExistenceOfParetoOptJME} states that $\mathcal{H}^{PO}_{\geq1}$ is compact. Furthermore, (\ref{eqBoundPn11}) implies that $\{1/p^{n}_{11}\}_{n\geq1}$ is bounded. Therefore, we can assume, passing to a subsequence if necessary, that
\begin{eqnarray}\label{eqLimQn}
\lim_{n\rightarrow\infty}q^{n}=q\in\mathcal{H}^{PO}_{\geq1}\subset{R}^{\infty}_{++},
\end{eqnarray}
and that
\begin{eqnarray}\label{eqLimInversePn11}
    \lim_{n\rightarrow\infty} \frac{1}{p^{n}_{11}}=\mu\geq0.
\end{eqnarray}
In particular, (\ref{eqLimQn}) implies
\begin{eqnarray}\label{eqLimQn2}
    \sum^{\infty}_{t=1}\frac{1}{\Vert q_{t}\Vert\sum_{h\in G^{2}_{t}}\Vert e^{h}_{t}\Vert}=+\infty.
\end{eqnarray}
Since $p^{n}\in\mathcal{H}^{PO}(n)$, $n\geq1$, the last part of Theorem \ref{theoExistenceOfParetoOptJME} implies
\begin{eqnarray*}
    s^{h^{*}_{n}}(p^{n}_{0},p^{n}_{1})+\sum_{h^{\prime}\in G_{0}(n)}s^{h^{\prime}}(p^{n}_{0},p^{n}_{1})=s^{h^{*}_{n}}(1,p^{n}_{11}q^{n}_{1})+1>\delta\biggr(e^{h^{*}_{n}}_{0}+\sum_{h^{\prime}\in G_{0}(n)}e^{h^{\prime}}_{0}\biggr)=2\delta,
\end{eqnarray*}
so that
\begin{eqnarray*}
    x^{h^{*}_{n}}_{0}(1,p^{n}_{11}q^{n}_{1})=\frac{1+p^{n}_{11}\Vert q^{n}_{1}\Vert\theta/n}{1+L_{1}/f(n)}<2(1-\delta).
\end{eqnarray*}
Rearranging terms, we obtain
\begin{eqnarray}\label{eqBoundsPn11}
    0<\frac{p^{n}_{11}}{n}<\frac{(1+L_{1}/f(n))2(1-\delta)-1}{\Vert q^{n}_{1}\Vert\theta}.
\end{eqnarray}
Since $\lim_{n\rightarrow\infty}\Vert q^{n}_{1}\Vert=\Vert q_{1}\Vert>0$, $2(1-\delta)-1>0$ (i.e., $\delta<1/2$) and
\begin{eqnarray*}
    \lim_{n\rightarrow\infty} f(n)=\lim_{n\rightarrow\infty}\frac{n L_{1}\sum_{h\in G_{1}}\Vert e^{h}_{1}\Vert}{2\theta(1+\varepsilon)}=+\infty,
\end{eqnarray*}
(\ref{eqBoundsPn11}) implies that $\{p^{n}_{11}/n\}_{n\geq1}$ is bounded. Then, (\ref{eqBoundPn11}) implies
\begin{eqnarray}
    \lim_{n\rightarrow\infty}x^{h^{*}_{n}}_{1}(p^{n}_{0},p^{n}_{1})&=&\lim_{n\rightarrow\infty}x^{h^{*}_{n}}_{1}(1,p^{n}_{11}q^{n}_{1})\nonumber\\
    &=&\lim_{n\rightarrow\infty}\frac{1+p^{n}_{11}\Vert q^{n}_{1}\Vert\theta/n}{f(n)+L_{1}}\frac{1}{p^{n}_{11}}\biggr(\frac{1}{q^{n}_{11}},\ldots,\frac{1}{q^{n}_{1L_{1}}}\biggr)\nonumber\\
    &< &\frac{1}{\lambda}\lim_{n\rightarrow\infty}\frac{1+p^{n}_{11}\Vert q^{n}_{1}\Vert\theta/n}{f(n)+L_{1}}\biggr(\frac{1}{q^{n}_{11}},\ldots,\frac{1}{q^{n}_{1L_{1}}}\biggr)\nonumber\\
    &=&0\label{eqHstarVanish}.
\end{eqnarray}
Next, since $q\in \mathcal{H}^{PO}_{\geq1}$, we have
\begin{eqnarray}\label{eqMarketClearingLimitq}
     \sum_{h\in G^{2}_{t-1}}x^{h}_{t}(q_{t-1},q_{t})+\sum_{h\in G^{2}_{t}}x^{h}_{t}(q_{t},q_{t+1})=\sum_{h\in G^{2}_{t-1}}e^{h}_{t}+\sum_{h\in G^{2}_{t}}e^{h}_{t},
\end{eqnarray}
for $t\geq2$. The market clearing equation for $\mathcal{E}(n)$ at $t=1$ is given by
\begin{eqnarray*}
    e^{h^{*}_{n}}_{1}+\sum_{h\in G^{2}_{0}} e^{h}_{1}+\sum_{h\in G^{2}_{1}} e^{h}_{1}
    &=&x^{h^{*}_{n}}_{1}(p^{n}_{0},p^{n}_{1})+\sum_{h^{\prime}\in G_{0}(n)} x^{h^{\prime}}_{1}(p^{n}_{0},p^{n}_{1})+\sum_{h\in G^{2}_{1}} x^{h}_{1}(p^{n}_{1},p^{n}_{2})\\
    &=&x^{h^{*}_{n}}_{1}(1,p^{n}_{1})+\sum_{h\in G^{2}_{0}} x^{h}_{1}(1/p^{n}_{11},q^{n}_{1})+\sum_{h\in G^{2}_{1}} x^{h}_{1}(q^{n}_{1},q^{n}_{2}),
\end{eqnarray*}
for $n\geq1$. Then, (\ref{eqHstarVanish}) and $\lim_{n\rightarrow+\infty}e^{h^{*}_{n}}_{1}=0$ imply
\begin{eqnarray}\label{eqMarketClearingLimitT1}
    \sum_{h\in G^{2}_{0}} x^{h}_{1}(\mu,q_{1})+\sum_{h\in G^{2}_{1}} x^{h}_{1}(q_{1},q_{2})=\sum_{h\in G^{2}_{0}} e^{h}_{1}+\sum_{h\in G^{2}_{1}} e^{h}_{1}.
\end{eqnarray}
Lemma \ref{lemmaSecondParetoOptEquilibria} and $p^{n}\in\mathcal{H}^{PO}(n)$, $n\geq1$, imply that the aggregate real savings of generation $G_{1}(n)$ from the economy $\mathcal{E}(n)$ satisfy 
\begin{eqnarray*}
    \delta&<& \frac{\sum_{h\in G_{1}(n)}s^{h}(p^{n}_{1},p^{n}_{2})}{\sum_{h\in G_{1}(n)} \Vert e^{h}_{1}\Vert}\\
    &=&\frac{p^{n}_{1} }{\Vert p^{n}_{1}\Vert\sum_{h\in G^{2}_{1}} \Vert e^{h}_{1}\Vert}\cdot \sum_{h\in G^{2}_{1}}(e^{h}_{1}-x^{h}_{1}(p^{n}_{1},p^{n}_{2}))\\
    &=&\frac{q^{n}_{1} }{\Vert q^{n}_{1}\Vert\sum_{h\in G^{2}_{1}} \Vert e^{h}_{1}\Vert}\cdot\biggr(x^{h^{*}_{n}}_{1}(p^{n}_{0},p^{n}_{1})-e^{h^{*}_{n}}_{1}+\sum_{h\in G^{2}_{0}} (x^{h}_{1}(1/p^{n}_{11},q^{n}_{1})- e^{h}_{1})\biggr),
\end{eqnarray*}
where the last equality is due to the market clearing equation in period $t=1$. Letting $n\rightarrow\infty$, we obtain
\begin{eqnarray*}
    \frac{\sum_{h\in G^{2}_{0}}q_{1}\cdot(x^{h}_{1}(\mu,q_{1})- e^{h}_{1})}{\Vert q_{1}\Vert \sum_{h\in G^{2}_{1}} \Vert e^{h}_{1}\Vert}\geq\delta.
\end{eqnarray*}
Then, Walras' law implies
\begin{eqnarray*}
     \frac{\sum_{h\in G^{2}_{0}}q_{1}\cdot(x^{h}_{1}(\mu,q_{1})- e^{h}_{1})}{\Vert q_{1}\Vert \sum_{h\in G^{2}_{1}} \Vert e^{h}_{1}\Vert}=\frac{\sum_{h\in G^{2}_{0}} \mu m^{h}}{\Vert q_{1}\Vert  \sum_{h\in G^{2}_{1}} \Vert e^{h}_{1}\Vert}=\frac{\mu}{\Vert q_{1}\Vert \sum_{h\in G^{2}_{1}} \Vert e^{h}_{1}\Vert}\geq\delta>0,
\end{eqnarray*}
and, therefore, $\mu>0$. Since $\mathcal{E}^{2}$ satisfies Assumption \ref{assCassCriterion}, (\ref{eqLimQn2}), (\ref{eqMarketClearingLimitq}) and (\ref{eqMarketClearingLimitT1}) allow us to conclude that $(\mu,q)\in\mathbb{R}^{\infty}_{++}$ is an optimal monetary equilibrium of $\mathcal{E}^{2}$.
\end{proof}


\printbibliography
\end{document}